\documentclass[letterpaper,11pt]{article}  

\usepackage{authblk}

\usepackage[margin=1in]{geometry}

\usepackage{times}
\usepackage{xcolor,boxedminipage}
\usepackage{soul}
\usepackage[utf8]{inputenc}
\usepackage[small]{caption}

\usepackage{booktabs}
\usepackage{bm}
\usepackage{tikz}
\usetikzlibrary{shapes}
\usetikzlibrary{calc}

\usepackage{boxedminipage}
\usepackage{xspace}
\usepackage{paralist}
\usepackage{floatrow}
\usepackage{multirow}

\newcommand{\sv}[1]{}
\newcommand{\lv}[1]{#1}

\usepackage{amsmath, amssymb}
\usepackage{amsthm}
\usepackage{color}
\usepackage{verbatim, xspace}
\usepackage{algorithm}
\usepackage{algpseudocode}

\usepackage[obeyFinal,colorinlistoftodos,textsize=tiny]{todonotes}

\theoremstyle{plain}
\newtheorem{THE}{Theorem}

\newtheorem{LEM}[THE]{Lemma}
\newtheorem{COR}[THE]{Corollary}
\newtheorem*{THE*}{Theorem}
\newtheorem*{fact*}{Fact}
\newtheorem{fact}[THE]{Fact}

\theoremstyle{definition}
\newtheorem{DEF}[THE]{Definition}

\newtheorem*{DEF*}{Definition}

\newtheorem{OBS}[THE]{Observation}

\newtheorem{CLM}{Claim}

\newcommand{\SB}{\{\,}
\newcommand{\SM}{\;{|}\;}
\newcommand{\SE}{\,\}}
\newcommand{\MMM}{\mathcal{M}}
\newcommand{\PPP}{\mathcal{P}}

\newcommand{\SSS}{\mathcal{S}}

\newcommand{\PE}{\textup{PE}}

\newcommand{\Nat}{\mathbb{N}}

\newcommand{\cc}[1]{{\mbox{\textnormal{\textsf{#1}}}}\xspace}  

\newcommand{\NP}{\cc{NP}}

\newcommand{\FPT}{\cc{FPT}}
\newcommand{\XP}{\cc{XP}}
\newcommand{\Weft}{{\cc{W}}}
\newcommand{\W}[1]{{\Weft}{{[#1]}}}
\newcommand{\paraNP}{\cc{paraNP}}

\newcommand{\cch}[1]{{\mbox{\textnormal{\textbf{\textsf{#1}}}}}\xspace}  

\newcommand{\FPTh}{\cch{FPT}}
\newcommand{\XPh}{\cch{XP}}
\newcommand{\Wefth}{{\cch{W}}}
\newcommand{\Wh}[1]{{\Wefth}{{\textbf{[#1]}}}}

\newcommand{\vc}{\mathsf{vc}}
\newcommand{\tw}{\mathsf{tw}}

\newcommand{\hy}{\hbox{-}\nobreak\hskip0pt}

\newcommand{\nn}{\mathbb{N}}

\newcommand{\bigoh}{\mathcal{O}}

\newcommand{\probfont}[1]{\textnormal{\textsc{#1}}}

\newcommand{\yes}{\textsc{Yes}}
\newcommand{\true}{\textsc{True}}
\newcommand{\GRAPHCON}{\textsc{TSS}}
\newcommand{\PSS}{\probfont{Partitioned Subset Sum}}
\newcommand{\MDPSS}{\probfont{MPSS}}
\newcommand{\SMDPSS}{\probfont{SMPSS}}
\newcommand{\SSP}{\probfont{Subset Sum}}

\newcommand{\mc}{\textsc{Partitioned Clique}}
\newcommand{\Gasp}{\probfont{Gasp}}
\newcommand{\sGasp}{\probfont{sGasp}}
\newcommand{\cGasp}{g\probfont{Gasp}}

\newcommand{\pref}{\succeq}
\newcommand{\prefs}{\succ}

\newcommand{\pbDef}[3]{%
\noindent
\begin{center}
\begin{boxedminipage}{0.98 \columnwidth}
#1\\[5pt]
\begin{tabular}{l p{0.73 \columnwidth}}
Input: & #2\\
Question: & #3
\end{tabular}
\end{boxedminipage}
\end{center}
}

\newcommand{\pbDefP}[4]{%
\noindent
\begin{center}
\begin{boxedminipage}{0.98 \columnwidth}
#1\\[5pt]
\begin{tabular}{l p{0.80 \columnwidth}}
Input: & #2\\
Parameter: & #3\\
Question: & #4
\end{tabular}
\end{boxedminipage}
\end{center}
}

\makeatletter
\newcommand\xleftrightarrow[2][]{%
  \ext@arrow 9999{\longleftrightarrowfill@}{#1}{#2}}
\newcommand\longleftrightarrowfill@{%
  \arrowfill@\leftarrow\relbar\rightarrow}
\makeatother

\title{Group Activity Selection with Few Agent Types}
\author[1]{Robert Ganian}
\author[2]{Sebastian Ordyniak}
\author[3]{C. S. Rahul}

\affil[1]{TU Wien, Vienna, Austria (rganian@gmail.com)}
\affil[2]{University of Sheffield, Sheffield, UK (sordyniak@gmail.com)}
\affil[3]{University of Warsaw, Warsaw, Poland (rahulcsbn@gmail.com)}
\date{}

\begin{document}

\maketitle

\begin{abstract}
  The Group Activity Selection Problem (GASP) models situations 
  where a group of agents needs to be distributed to a set of
  activities while taking into account preferences of the agents
  w.r.t.\ individual activities and activity sizes. The problem, along
  with its two previously proposed variants sGASP and gGASP, has been studied in the parameterized complexity setting with various parameterizations, such as \emph{number of agents}, \emph{number of activities} and \emph{solution size}. However, the complexity of the problem parameterized by the \emph{number of types of agents}, a parameter motivated and proposed already in the paper that introduced GASP, has so far remained open.
  
In this paper we establish the complexity map for GASP, sGASP and gGASP when the number of types of agents is the parameter. Our positive results, consisting of one fixed-parameter algorithm and one XP algorithm, rely on a combination of novel Subset Sum machinery (which may be of general interest) and identifying certain \emph{compression steps} which allow us to focus on solutions which are ``acyclic''. These algorithms are complemented by matching lower bounds, which among others answer an open question of Gupta, Roy, Saurabh and Zehavi (2017). In this direction, the techniques used to establish W[1]-hardness of sGASP are of particular interest: as an intermediate step, we use Sidon sequences to show the W[1]-hardness of a highly restricted variant of multi-dimensional Subset Sum, which may find applications in other settings as well.
%
%

\end{abstract}


\section{Introduction}
\label{sec:intro}

In this paper we consider the \textsc{Group Activity Selection
  Problem} (\Gasp{}) together with its two most
prominent variants, the \textsc{Simple Group Activity Selection Problem}
(\sGasp{}), and the \textsc{Group Activity Selection Problem with
  Graph Structure} (\cGasp{})~\cite{DarmannEKLSW12,IgarashiPE17}. Since their
  introduction, these problems have become the focal point of extensive research~\cite{Darmann15,DarmannDockerDornLangSchneckenburger17,darmann2017,LeeW17GASP,GuptaRoySaurabhZehavi17,IBR17}.
In \Gasp{} one is given a set of agents, a set of activities, and a set of preferences for
each agent in the form of a complete transitive relation (also called the
\emph{preference list}) over the set of pairs
consisting of an activity $a$ and a number $s$, expressing the
willingness of the agent to participate in the activity $a$ if it has
$s$ participants. The aim is to find
a ``good'' assignment from agents to activities subject to certain
\emph{rationality} and \emph{stability} conditions.
Specifically, an assignment is individually rational if agents that are assigned to
an activity prefer this outcome over not being assigned to any
activity, and an assignment is (Nash) stable if every agent prefers
its current assignment over moving to any other activity.
In this way GASP captures a wide range of real-life
situations such as event organization and work delegation.

\sGasp{} is a simplified variant of \Gasp{} where the preferences of
agents are expressed in terms of \emph{approval sets} containing
(activity,size) pairs  instead of preference lists. In essence
\sGasp{} is \Gasp{} where each preference list has only two
equivalence classes: the class of the approved (activity,size)
pairs
(which contains all pairs
that are preferred  over not being assigned
to any activity), and the class of disapproved (activity,size) pairs
(all possible remaining pairs). 
On the other hand, \cGasp{} is a generalization of \Gasp{} where one is additionally
given an undirected graph (network) on the set of all agents that can
be employed to model for instance acquaintanceship or
physical distance between agents. Crucially, in \cGasp{} one only considers
assignments for which the subnetwork induced by all agents assigned to
some activity is connected. 
Note that 
if the network
forms a complete graph, then \cGasp{} is equivalent to the underlying
\Gasp{} instance.


\paragraph{Related Work.}
\sGasp{}, \Gasp{}, and \cGasp{}, are
known to be \NP-complete even in very restricted
settings~\cite{DarmannEKLSW12,IgarashiPE17,GuptaRoySaurabhZehavi17,IBR17}. It is therefore natural
to study these problems through the lens of parameterized complexity~\cite{DowneyFellows13,CyganFKLMPPS15}.
%
Apart from parameterizing by the solution size (i.e., the number of agents
assigned to any activity in a solution)~\cite{LeeW17GASP}, the perhaps 
%
most prominent parameterizations
thus far have been the number of activities, the number of agents, and in the case of
\cGasp{} structural parameterizations tied to the structure of the network such as
treewidth~\cite{DarmannEKLSW12,IgarashiPE17,GuptaRoySaurabhZehavi17,EibenGanianOrdyniakActivity18}.
Consequently, the parameterized complexity of all three variants of \Gasp{} w.r.t.\
the number of activities and/or the number of agents 
is now almost completely understood. 

Computing a stable assignment for a given instance of \Gasp{} is known to be \W{1}\hy hard and contained in \XP{}
parameterized by either the number of
activities~\cite{DarmannEKLSW12,IgarashiPE17} or the number of
agents~\cite{IgarashiPE17} and known to be fixed-parameter tractable
parameterized by both parameters~\cite{IgarashiPE17}.
Even though it has never been explicitly stated, the same results also
hold for \cGasp{} when parameterizing by the number of agents as well
as when using both parameters. This is because both the
\XP{} algorithm for the number of agents as well as the fixed-parameter algorithm
for both parameters essentially brute-force over every possible
assignment and are hence also able to find a solution for
\cGasp{}. Moreover, the fact that \cGasp{} generalizes \Gasp{} implies
that the \W{1}\hy hardness result for the number of agents also carries over to
\cGasp{}. 
On the other hand, if we consider the number of activities as a
parameter then \cGasp{} turns out to be harder than \Gasp{}: Gupta et al.~(\cite{GuptaRoySaurabhZehavi17}) showed that
\cGasp{} is \NP-complete even when restricted to instances with a single activity.
The hardness of \cGasp{} has inspired a series of tractability
results~\cite{GuptaRoySaurabhZehavi17,IgarashiPE17} obtained by
employing additional restrictions on the structure of the network.
One prominent result in this direction has been recently obtained by Gupta et al.~(\cite{GuptaRoySaurabhZehavi17}), showing
that \cGasp{} is fixed-parameter tractable parameterized by the number
of activities if the network has constant treewidth. 
For \sGasp{}, it was recently shown that the problem is also \W{1}-hard when parameterized by the number of activities~\cite{EibenGanianOrdyniakActivity18}, and hence the only small gap left 
was the complexity of this problem parameterized by the number of agents.
For completeness, we resolve this question in our concluding remarks by giving a fixed-parameter algorithm.

Already with the introduction of \Gasp{}~\cite{DarmannEKLSW12} the authors argued that 
instead of putting restrictions on the total number
of agents, which can be very large in general, it might
be much more useful to consider the number of distinct types of
agents. It is easy to imagine a setting
with large groups of agents that share the same preferences
(for instance due to inherent limitations of how preferences are
collected). In contrast to the related parameter number of activity
types, where it is known
that \sGasp{} remains \NP-complete even for a
constant number of activity types~\cite{DarmannEKLSW12}, the complexity of the problems
parameterized by the number of agent types (with or without restricting the number of activities) has remained wide open thus far.

\paragraph{Our Results.}
In this paper we obtain a complete classification of the complexity of \Gasp{} and its variants \sGasp\ and \cGasp\ when parameterized by the number of agent types ($t$) alone, and also when parameterized by $t$ plus the number of activities ($a$). In particular, for each of the considered problems and parameterizations, we determine whether the problem is in \FPT, or \W{1}-hard and in \XP, or para\NP-hard. One distinguishing feature of our lower- and upper-bound results is that they make heavy use of novel Subset-Sum machinery. Below, we provide a high-level summary of the individual results presented in the paper.

\begin{itemize}
\item[\textbf{Result 1.}] \textbf{\sGasp{} is fixed-parameter tractable when parameterized by $t+a$.}
\end{itemize}

This is the only fixed-parameter tractability result presented in the paper, and is essentially tight: it was recently shown that \sGasp{} is \W{1}-hard when parameterized by $a$ alone~\cite{EibenGanianOrdyniakActivity18}, and the \W{1}-hardness of the problem when parameterized by $t$ is obtained in this paper. Our first step towards obtaining the desired fixed-parameter algorithm for \sGasp{} is to show that every YES-instance has a solution which is \emph{acyclic}---in particular, a solution with no cycles formed by interactions between activities and agent types (captured in terms of the \emph{incidence graph} $G$ of an assignment). This is proved via the identification of certain \emph{compression steps} which can be applied on a solution in order to remove cycles. 

Once we show that it suffices to focus on acyclic solutions, we branch over all acyclic incidence graphs (i.e., all acyclic edge sets of $G$); for each such edge set, we can reduce the problem of determining whether there exists an assignment realizing this edge set to a variant of \textsc{Subset Sum} embedded in a tree structure. The last missing piece is then to show that this problem, which we call \textsc{Tree Subset Sum}, is polynomial-time tractable; this is done via dynamic programming, whereas each step boils down to solving a simplified variant of \textsc{Subset Sum}. 

\begin{itemize}
\item[\textbf{Result 2.}] \textbf{\sGasp{} is \W{1}-hard when parameterized by $t$.}
\end{itemize}
Our second result complements Result 1. As a crucial intermediate step towards Result 2, we 
%
obtain the \W{1}-hardness of a variant of \textsc{Subset Sum} with three distinct ``ingredients'':
\begin{enumerate}
\item \label{prop1} \textbf{Partitioning}: items are partitioned into sets, and precisely one item must be selected from each set,
\item \label{prop2} \textbf{Multidimensionality}: each item is a $d$-dimensional vector ($d$ being the parameter) where the aim is to hit the target value for each component, and
\item \label{prop3} \textbf{Simplicity}: each vector contains precisely one non-zero component.
\end{enumerate}
We call this problem \textsc{Simple Multidimensional Partitioned Subset Sum} (\textsc{SMPSS}). Note that \textsc{SMPSS} is closely related to \textsc{Multidimensional Subset Sum} (\textsc{MSS}), which (as one would expect) merely enhances \textsc{Subset Sum} via Ingredient~\ref{prop2}. \textsc{MSS} has recently been used to establish \W{1}-hardness for parameterizations of \textsc{Edge Disjoint Paths}~\cite{GanianOrdyniakRamanujan17} and \textsc{Bounded Degree Vertex Deletion}~\cite{GanianKO18}. However, Ingredient~\ref{prop1} and especially Ingredient~\ref{prop3} are critical requirements for our reduction to work, and establishing the \W{1}-hardness of \textsc{SMPSS} was the main challenge on the way towards the desired lower-bound result for \sGasp{}. Since \textsc{MSS} has already been successfully used to obtain lower-bound results and \textsc{SMPSS} is a much more powerful tool in this regard, we believe that \textsc{SMPSS} will find applications in establishing lower bounds for other problems in the future.

\begin{itemize}
\item[\textbf{Result 3.}] \textbf{\Gasp{} is in \XP\ when parameterized by $t$.}
\end{itemize}
This is the only \XP\ result required for our complexity map, as it
implies \XP\ algorithms for \sGasp{} parameterized by $t$ and for
\Gasp{} parameterized by $t+a$. We note that the techniques used to
obtain Result 3 are disjoint from those used for Result 1; in particular, our first
step is to reduce \Gasp{} parameterized by $t$ to solving ``\XP-many''
instances of \sGasp{} parameterized by $t$. This is achieved by showing
that once we know a ``least preferred alternative'' for every agent
type that is active in an assignment, then the \Gasp{} instance
becomes significantly easier---and, in particular, can be reduced to a
(slightly modified version of) \sGasp{}. It is interesting to note that
the result provides a significant conceptual improvement over the known brute
force algorithm for \Gasp{} parameterized by the number of agents which
enumerates all possible assignments of agents to
activities~\cite[Theorem 3]{IgarashiBredereckElkind17arv} (see also~\cite{IBR17}): instead of
guessing an assignment for all agents, one merely needs to guess a least
preferred alternative for every agent type.

The second part of our journey towards Result 3 focuses on obtaining
an \XP\ algorithm for \sGasp{} parameterized by $t$. This algorithm
has two components. Initially, we show that in this setting one can
reduce \sGasp{} to the problem of finding an assignment which is
individually rational (i.e., without the stability condition) and
satisfies some additional minor properties. To find
such an assignment, we once again make use of \textsc{Subset Sum}: in
particular, we develop an \XP\ algorithm for the \textsc{MPSS} problem
(i.e., \textsc{Subset Sum} enhanced by ingredients~\ref{prop1}
and~\ref{prop2}) and apply a final reduction from finding an
individually rational assignment to \textsc{MPSS}.

\begin{itemize}
\item[\textbf{Result 4.}] \textbf{\Gasp{} is \W{1}-hard when parameterized by $t+a$.}
\item[\textbf{Result 5.}] \textbf{\cGasp{} is \W{1}-hard when parameterized by $t+a$ and the \emph{vertex cover number}~\cite{FellowsLokshtanovMisra08} of the network.}
\end{itemize}

The final two results are hardness reductions which represent the last
pieces of the presented complexity map. Both are obtained via
reductions from \textsc{Partitioned Clique} (also called
\textsc{Multicolored Clique} in the
literature~\cite{CyganFKLMPPS15}), and both reductions essentially use
$k+\binom{k}{2}$ activities whose sizes encode the vertices and edges
forming a $k$-clique. The main
challenge lies in the design of (a bounded number of) agent types
whose preference lists ensure that the chosen vertices are indeed
endpoints of the chosen edges. The reduction for \cGasp{} then
becomes even more involved, as it can only employ a limited number of connections
between the agents in order to ensure that vertex cover of the network is bounded.

We note that Result 5 answers an open question raised by Gupta, Roy, Saurabh and Zehavi~\cite{GuptaRoySaurabhZehavi17}, who showed that \cGasp{} is fixed-parameter tractable parameterized by the number of activities
if the network has constant treewidth and wondered whether their
result can be improved to a more efficient fixed-parameter algorithm
parameterized by the number of activities and treewidth. In this sense, our hardness result
represents a substantial shift of the boundaries of (in)tractability: in addition to the
setting of Gupta et al., it also rules out the use of agent types as a
parameter and replaces treewidth by the more restrictive vertex cover number.

An overview of our results in the context of related work is provided in Table~\ref{tbl:intro-res}. 

\begin{table}[ht]
  \begin{tabular}{r|c|cc}
    & Parameterization & Lower Bound & Upper Bound \\\toprule
    \sGasp{} & \multirow{3}{*}{$n$} & \multicolumn{2}{c}{{\color{white}$^1$}\FPTh{}$^0$} \\
    \Gasp{} & & {\color{white}$_a$}\W{1}$_a$ & \XP \\
    \cGasp{} & & \W{1} & {\color{white}$_a$}\XP{}$_a$ \\ \midrule

        \sGasp{} & \multirow{3}{*}{$a$} & {\color{white}$_b$}\W{1}$_b$  & \XP{} \\
    \Gasp{} & & \W{1} &  {\color{white}$_a$}\XP{}$_a$ \\
    \cGasp{} & & \multicolumn{2}{c}{{\color{white}$_c$}\paraNP{}$_c$} \\\midrule

    \sGasp{} & \multirow{3}{*}{$n+a$} & \multicolumn{2}{c}{\multirow{3}{*}{\FPT{}$_d$}}\\
    \Gasp{} & & \\
    \cGasp{} & & \\\midrule

    \sGasp{} & \multirow{3}{*}{$t$} & {\color{white}$^2$}$\Wh{1}^2$ & \XPh{}  \\
    \Gasp{} & & \Wh{1} & {\color{white}$^3$}\XPh{}$^3$\\
    \cGasp{} & & \multicolumn{2}{c}{{\color{white}$_c$}\paraNP{}$_c$} \\\midrule
        
    \sGasp{} & \multirow{3}{*}{$t+a$} & \multicolumn{2}{c}{{\color{white}$^1$}\FPTh{}$^1$}\\
    \Gasp{} & & {\color{white}$^4$}\Wh{1}$^4$ & \XPh{}\\
    \cGasp{} & & \multicolumn{2}{c}{{\color{white}$_c$}\paraNP{}$_c$}\\\midrule
    
    \cGasp{} & $t+a+\vc$ (or $t+a+\tw$) & {\color{white}$^5$}\Wh{1}$^5$ & \XP{}$_c$
  \end{tabular}\vspace{-0.2cm}
  \caption[]{Lower and upper bounds for \sGasp{}, \Gasp{}, and \cGasp{} parameterized
  by the number of agents ($n$), agent types ($t$), and the number
  of activities ($a$). In the case of
    $\cGasp{}$, also the parameters vertex cover number ($\vc$) and
    treewidth ($\tw$) of the network are considered. Entries in bold are shown in this paper, and the numbers $1$ to $5$ in the upper index are used to identify results 1 to 5. The result marked with $^0$ is provided in the concluding remarks.\\
    References: $a$ is~\cite{IgarashiPE17}, $b$ is~\cite{EibenGanianOrdyniakActivity18}, $c$ is~\cite{GuptaRoySaurabhZehavi17}, $d$ is folklore.}
  \label{tbl:intro-res}
\end{table}

\paragraph{Organization of the Paper.}
After introducing the required preliminaries in Section~\ref{sec:prelims}, we present all of our Subset Sum machinery in the dedicated Section~\ref{sec:ssm}. Each subsequent Section $i\leq 8$ then focuses on obtaining Result $i-3$.

\section{Preliminaries}
\label{sec:prelims}

For an integer $i$, we let $[i]=\{1,2,\dots,i\}$ and $[i]_0=[i]\cup\{0\}$.
We denote by $\Nat$ the set of natural numbers, by $\Nat_0$ the set
$\Nat \cup \{0\}$. For a set $S$ and an integer $k$, we denote by
$S^k$ and $2^S$ the set of all $k$ dimensional vectors over $S$ and the
set of all subsets of $S$, respectively. 
For a vector $\bar{p}$ of integers, we use $|\bar{p}|$ to denote the
sum of its elements.

We refer to the handbook by Diestel~(\cite{Diestel12}) for
standard graph terminology. The \emph{vertex cover number} of a graph $G$ is the size of a minimum vertex cover of $G$.

\subsection{Parameterized Complexity}

In parameterized
algorithmics~\cite{DowneyFellows13,CyganFKLMPPS15,Niedermeier06}
the run-time of an algorithm is studied with respect to a parameter
$k\in\nn_0$ and input size~$n$.
The basic idea is to find a parameter that describes the structure of
the instance such that the combinatorial explosion can be confined to
this parameter.
In this respect, the most favourable complexity class is \FPT (\textit{fixed-parameter tractable})
which contains all problems that can be decided by an algorithm
running in time $f(k)\cdot n^{\bigoh(1)}$, where $f$ is a computable
function.
Algorithms with this running time are called \emph{fixed-parameter algorithms}. A less favourable outcome is an \XP \emph{algorithm}, which is an algorithm running in time $\mathcal O(n^{f(k)})$; problems admitting such algorithms belong to the class \XP.
\sv{
To obtain our lower bounds, we will need the notion of a
  \emph{parameterized reduction} and the complexity class
  $\W{1}$. Generally speaking, a parameterized reduction is a variant
  of the standard polynomial-time reduction which retains bounds on
  the parameter, and $\W{1}$\hy hardness rules out the existence of fixed-parameter algorithms under the Exponential Time Hypothesis~\cite{ImpagliazzoPaturiZane01}.
  }

To obtain our lower bounds, we will need the notion of a parameterized reduction.
Formally, a {\em parameterized problem\/} is a subset of $\Sigma^*\times\nn$, 
where $\Sigma$ is the input alphabet.
Let $L_1\subseteq \Sigma_1^*\times\nn$ and $L_2\subseteq \Sigma_2^*\times\nn$ be parameterized problems.
A \textit{parameterized reduction} (or FPT-reduction) from $L_1$ to $L_2$ is a mapping $P:\Sigma_1^*\times\nn\rightarrow\Sigma_2^*\times\nn$ such that
\begin{inparaenum}[(i)]
  \item $(x,k)\in L_1$ iff $P(x,k)\in L_2$,
  \item the mapping can be computed by an FPT-algorithm w.r.t.\ parameter $k$, and
  \item there is a computable function $g$ such that $k'\leq g(k)$, where $(x',k')=P(x,k)$.
\end{inparaenum}

Finally, we introduce the complexity class used to describe our lower bounds.
The class $\W{1}$ captures parameterized intractability and contains
all problems that are FPT-reducible to \probfont{Independent Set} 
(parameterized by solution size).

\subsection{Group Activity Selection.}
The task in the \textsc{Group Activity Selection Problem} (\textsc{Gasp}) is to compute a \textit{stable assignment} $\pi$ from a given set $N$ of
\textit{agents} to a set $A$ of \textit{activities}, where each agent
participates in at most one activity in $A$. The assignment
$\pi$ is (Nash) stable if and only if it is \textit{individually
  rational} and no agent has an \textit{NS-deviation} to any other activity (both of these
stability rules are defined in the next paragraph). We use a
dummy activity $a_\emptyset$ to capture all those agents that do not
participate in any activity in $A$ and denote by $A^*$ the set $A\cup\{a_\emptyset\}$.  Thus, an assignment $\pi$ is a mapping from $N$ to $A^*$, and for an activity $a\in A$ we use $\pi^{-1}(a)$ to denote the set of agents assigned to $a$ by $\pi$; we set $|\pi^{-1}(a_{\emptyset})|=1$ if there is at least one agent assigned to $a_\emptyset$ and $0$ otherwise.    

The set $X$ of \emph{alternatives} is defined as $X=(A\times [|N|])\cup\{(a_\emptyset,1)\}$. Each agent is associated with its own \textit{preferences} defined on the set $X$. 
In the case of the standard \Gasp{} problem, an instance $I$ is of the form $(N,A,(\pref_{n})_{n\in N})$ where each agent $n$ is associated with a complete transitive preference relation (list) $\pref_n$ over the set $X$. An assignment $\pi$ is \emph{individually  rational} if for every agent $n \in N$ it holds that if $\pi(n)=a$ and $a \neq
a_\emptyset$, then $(a,|\pi^{-1}(a)|)\pref_n
(a_\emptyset,1)$ (i.e., $n$ weakly prefers staying in $a$ over moving to $a_\emptyset$). 
An agent $n$ where $\pi(n)=a$  is defined
to have an \textit{NS-deviation} to a different
activity $a'$ in $A$ if $(a',|\pi^{-1}(a')|+1) \prefs_n
(a,|\pi^{-1}(a)|)$ (i.e., $n$ prefers moving to an activity $a'$ over staying in $a$). 
The task in \Gasp{} is to compute a stable assignment.

\cGasp{} is defined analogously to \Gasp{}, however where one is additionally given a set $L$ of
\emph{links} $L \subseteq \SB \{n,n'\}\SM
n,n'\in N \land n\neq n'\SE$ between the agents on the input; specifically, $L$ can be viewed as a set of undirected edges and $(N,L)$ as a simple undirected graph. In \cGasp{}, the task is to find an assignment $\pi$ which is not only stable but also connected; formally, for every $a \in A$ the set of agents $\pi^{-1}(a)$ induces a connected subgraph of $(N,L)$. Moreover, an agent $n \in N$ only has an NS-deviation to some activity $a\neq \pi(n)$ if (in addition to the conditions for NS-deviations defined above) $n$ has an edge to at least one agent in $\pi^{-1}(a)$.

In \sGasp{}, an instance $I$ is of the form $(N,A,(P_n)_{n \in N})$, where each agent has an \emph{approval set} $P_{n}\subseteq X\setminus \{(a_\emptyset,1)\}$ of preferences (instead of an ordered preference list).
We denote by $P_n(a)$ the set $\SB i \SM (a,i)\in P_n\SE$ for an activity $a \in A$. An assignment $\pi : N \rightarrow A^*$ is said to be \emph{individually rational} if every agent $n \in N$ satisfied the following: if $\pi(n)=a$ and $a \neq
a_\emptyset$, then $|\pi^{-1}(a)| \in P_n(a)$. Further, an agent $n \in N$ where 
$\pi(n)=a_\emptyset$,  is said to have an \emph{NS-deviation}  to an
activity $a$ in $A$ if $(|\pi^{-1}(a)|+1) \in P_n(a)$.

We now introduce the notions and definitions required for our
main parameter of interest, the ``number of agent types''.
We say that two agents $n$ and $n'$ in $N$ have the same \emph{agent type} if they have the same preferences. To be specific,  
$P_n=P_{n'}$ for \sGasp{} and $\pref_n=\pref_{n'}$ for \Gasp{} and \cGasp{}. In the case of \sGasp{} and \Gasp{} $n$ and $n'$ are indistinguishable, while in \cGasp{} $n$ and $n'$ can still have different links to other agents.
For a subset $N' \subseteq N$, we denote by $T(N')$
the set of agent types occurring in $N'$. Note that this notation
requires that the instance is clear from the context. If this is not
the case then we denote by $T(I)$ the set $T(N)$ if $N$ is the set of
agents for the instance $I$ of \sGasp{}, \Gasp{}, or \cGasp{}. 

For every agent type $t \in T(I)$, we denote by $N_t$ the
subset of $N$ containing all agents of type $t$; observe that $\SB N_t
\SM t \in T(I) \SE$ forms a partition of $N$. 
For an agent type $t \in T(I)$ we denote by $P_t$ (\sGasp{}) or $\pref_t$ (\Gasp{}) the preference list assigned to all agents of type $t$ and we use $P_t(a)$ (for an activity $a \in A$) to denote $P_t$ restricted to activity $a$, i.e., $P_t(a)$ is equal to
$P_n(a)$ for any agent $n$ of type $t$.
For an assignment $\pi : N \rightarrow A^*$, $t \in T(I)$,
and $a \in A$ we denote by $\pi_{t,a}$ the set $\SB n
\SM n \in N_t \land \pi(n)=a\SE$ and by $\pi_t$ the set $\bigcup_{a \in
  A}\pi_{t,a}$.
Further, $\pi(t)$ is the set of all activities that have at least
one agent of type $t$ participating in it.
We say that $\pi$ is a perfect assignment for some agent type $t \in T(I)$
if $\pi(n)\neq a_\emptyset$ for every $n \in N_t$. We denote by
$\PE(I,\pi)$ the subset of $T(I)$ consisting of all agent types that are
perfectly assigned by $\pi$, and say that $\pi$ is a \emph{perfect assignment} if $\PE(I,\pi)=T(I)$.

One notion that will appear throughout the paper is that of \emph{compatibility}:
for a subset $Q \subseteq T(I)$, we say that $\pi$ is
\emph{compatible} with $Q$ if $\PE(I,\pi)=Q$.
We conclude this section with a technical lemma
which provides a preprocessing procedure that will be used as a basic tool for obtaining our algorithmic results. In particular, Lemma~\ref{lem:fpt-prepro} allows us 
to reduce the problem of
computing a stable assignment for a \sGasp{} instance compatible with $Q$ to the problem
of finding an individually rational assignment.

\begin{LEM}
\label{lem:fpt-prepro}
  Let $I=(N,A,(P_n)_{n \in N})$ be an instance of \sGasp{} and
  $Q \subseteq T(I)$. Then in time $\bigoh(|N|^2 |A|)$
  one can compute an instance $\gamma(I,Q)=(N,A,(P_n')_{n
    \in N})$ and $A_{\neq \emptyset}(I,Q)\subseteq A$ with the
    following property: for every assignment $\pi : N \rightarrow A^*$ that is compatible
  with $Q$, it holds that $\pi$ is stable for $I$ if and
  only if $\pi$ is individually rational for $\gamma(I,Q)$ and 
  $\pi^{-1}(a)\neq \emptyset$ for every $a \in A_{\neq
    \emptyset}(I,Q)$.
\end{LEM}
\begin{proof}
  Let $\pi: N \rightarrow A^*$ be an assignment that is compatible
  with $Q$. Then there is an agent of type $t\in T(N)$ assigned to 
  $a_\emptyset$ if and only if $t \notin Q$. Hence $\pi$ is stable for $I$
  if and only if $\pi$ is individually rational and additionally it
  holds that for every agent type $t \in T(N) \setminus Q$ and every
  activity $a \in A$, $|\pi^{-1}(a)|+1 \notin P_t(a)$. This naturally leads us to a
  certain set of ``forbidden'' sizes for activities, and we will obtain
  the desired instance $\gamma(I,Q)$ by simply removing all tuples 
  from all preference lists that would allow activities to reach a forbidden size.
  Formally, we obtain the desired instance $\gamma(I,Q)$ from $I$ removing all tuples
  $(a,n)$ from every preference list $P_t$, where $t \in T(N)$ such
  that there is an agent type $t' \in T(N) \setminus Q$ with $n+1 \in
  P_{t'}(a)$.  
  This construction prevents the occurrence of all forbidden sizes of activities 
  \emph{except} for forbidding activities of size $0$; that is where we use the
  set $A_{\neq \emptyset}(I,Q)$.
  Formally, the set $A_{\neq \emptyset}(I,Q)$ consists of
  all activities $a$ such that there is an agent type $t \in T(N)
  \setminus Q$ with $1 \in P_t(a)$.  
  It is now straightforward to
  verify that $\gamma(I,Q)$ and $A_{\neq \emptyset}(I,Q)$ satisfy the claim
  of the lemma.

  Finally the running time of $\bigoh(|N|^2|A|)$ for the algorithm can be achieved as follows.
  In a preprocessing step we first compute for every activity $a \in
  A$ the set of all forbidden numbers, i.e., 
  the set of all numbers $n$ such that there is an agent type $t
  \in T(N) \setminus Q$ with $(a,n+1) \in P_{t'}(a)$. For every
  activity $a$, we store the resulting set of numbers in such a way
  that determining whether a number $n$ is contained in the set for
  activity $a$ can be achieved in constant time; this can for instance
  be achieved by storing the set for each activity $a$ as a Boolean array with $|N|$
  entries, whose $i$-th entry is \true{} if and only if $i$ is contained in
  the set of numbers for $a$. This preprocessing step takes time at
  most $\bigoh(|N|^2 |A|)$ and after it is completed we
  can use the computed sets to test for every agent type $t
  \in T(N)$, every activity $a \in A$, and every $n \in P_t(a)$, whether
  there is an agent type $t' \in T(N) \setminus Q$ such that $(a,n+1)
  \in P_{t'}(a)$ in constant time. If so we remove $n$ from $P_t(a)$,
  otherwise we continue. This shows that $\gamma(I,Q)$ can be computed in
  $\bigoh(|N|^2\cdot |A|)$ time. The computation of $A_{\neq
    \emptyset}$ only requires to check for every activity $a \in A$
  whether $1$ is contained in the set of forbidden numbers for $a$; if
  so $a$ is contained in $A_{\neq \emptyset}$ and otherwise it is
  not. After preprocessing, this can be achieved in time $\bigoh(|A|)$.
\end{proof}

\section{Subset Sum Machinery}\label{sec:ssm}

In this section we introduce the subset sum machinery required for our algorithms and lower bound results. 
In particular, we introduce three variants of \textsc{Subset Sum}, obtain algorithms for two of them, and provide a \W{1}\hy hardness result for the third.
We note that it may be helpful to read the following three subsections in the context of the individual sections where they are used: in particular, Subsection~\ref{ssec:tss} is used to obtain Result~1 (Section~\ref{sec:res1}), Subsection~\ref{ssec:mpss} is used as a preprocedure for Result~3 (Section~\ref{sec:res3}) and Subsection~\ref{ssec:smpss} (which is by far the most difficult of the three) is a crucial step in the reduction used for Result~2 (Section~\ref{sec:res2}).


\subsection{Tree Subset Sum}\label{ssec:tss}

Here we introduce a useful generalization of
\SSP, for which we obtain polynomial-time tractability under the assumption that
the input is encoded in unary.
Intuitively, our problem asks us to assign values to edges 
while meeting a simple criterion on the values of edges incident to each vertex.

\pbDef{\textsc{Tree Subset Sum} (\GRAPHCON)}{A vertex-labeled undirected tree $T$ with labeling
  function $\lambda : V(T) \rightarrow 2^{\Nat}$.}
{Is there an assignment $\alpha: E(T) \rightarrow \Nat$ such
  that for every $v \in V(T)$ it holds that $\sum_{e\in E(T) \land v \in
    e}\alpha(e)\in \lambda(v)$.}

Let us briefly comment on the relationship of \GRAPHCON{} with \SSP{}.
Recall that given a set $S$ of natural numbers and a natural number
$t$, the \SSP{} problem
asks whether there is a subset $S'$ of $S$ such that $\sum_{s \in
  S'}s=t$. One can easily construct a simple instance
$(G,\lambda)$ of \GRAPHCON{} that is equivalent to a given instance
$(S,t)$ of \SSP{} as follows. $G$ consists of a star having
one leaf $l_s$ for every $s \in S$ with $\lambda(l_s)=\{0,s\}$ and
$\lambda(c)=\{t\}$ for the center vertex $c$ of the star. 
Given this
simple reduction from \SSP{} to \GRAPHCON{} it becomes clear that
\GRAPHCON{} is much more general than \SSP{}.
In particular, instead
of a star \GRAPHCON{} allows for the use of an arbitrary tree
structure and moreover one can use arbitrary subsets of natural
numbers to specify the constrains on the vertices.
The above reduction in combination with the fact that \SSP{} is weakly
NP-hard implies that \GRAPHCON{} is also weakly NP-hard. 

In the remainder of this section we will
show that \GRAPHCON{} (like \SSP{}) can be solved in
polynomial-time if the input is given in unary.

Let $I=(T,\lambda)$ be an instance of \GRAPHCON. We denote by
$\max(I)$ the value of the maximum number occurring in any vertex
label. The main idea behind our algorithm for \GRAPHCON{} is to apply 
leaf-to-root dynamic programming. 
In order to execute our dynamic programming procedure,
we will need to solve a special case of \GRAPHCON{} which we call \PSS{}; 
this is the problem that will later arise when computing
the dynamic programming tables for \GRAPHCON.
\sv{In the \PSS{} problem one is given a target set $R$ of natural numbers and $\ell$ source sets
$S_1, \dotsc S_\ell$ of natural numbers and the aim is to
compute the set $S$ of all natural numbers $s$ such that there are
$s_1,\dotsc,s_\ell$, where $s_i \in S_i$ for every $i$ with $1 \leq i
\leq \ell$, satisfying $(\sum_{1 \leq i \leq \ell}s_i)+s \in R$.}

\pbDef{\PSS}{A target set $R$ of natural numbers and $\ell$ source sets
  $S_1, \dotsc S_\ell$ of natural numbers.}
{Compute the set $S$ of all natural numbers $s$ such that there are
  $s_1,\dotsc,s_\ell$, where $s_i \in S_i$ for every $i$ with $1 \leq i
  \leq \ell$, satisfying $(\sum_{1 \leq
    i \leq \ell}s_i)+s \in R$.}
For an instance $I=(R,S_1,\dotsc,S_\ell)$ of \PSS{}, we denote by
$\max(I)$ the value of the maximum number occurring in $R$.
\begin{LEM}
  \label{lem:solve-pss}
  An instance $I=(R,S_1,\dotsc,S_\ell)$ of \PSS{} can be solved in
  time $\bigoh(\ell \cdot \max(I)^2)$.
\end{LEM}
  

\begin{proof}
  Here we also use a dynamic programming approach similar to the approach
  used for the well-known \SSP{} problem~\cite{GareyJohnson79}. Let
  $I=(R,S_1,\dotsc,S_\ell)$ be an instance of \PSS{}.

  We first
  apply a minor modification to the instance which will allow us to
  provide a cleaner presentation of the algorithm.
  Namely, let $P_0,
  P_1,\dotsc,P_\ell$ be sets of integers defined as
  follows: $P_0=R$, and for every $i\in [\ell]$,  
  we set $P_i=\SB -s \SM s\in S_i \SE$. Then the solution $S$ for
  $I$ is exactly the set of all natural numbers $n$ for which there are
  $p_0,\dotsc,p_{\ell}$ with $p_i \in P_i$ for every $i$ with $0 \leq i
  \leq \ell$ such that $\sum_{0\leq i \leq \ell}p_i=n$.
  
  In order to compute the solution $S$ for $I$ (employing the
  above characterization for $S$), we
  compute a table $D$ having one binary entry $D[i,n]$ for every $i$
  and $n$ with $0 \leq i
  \leq \ell$ and $0\leq n \leq \max(I)$ such that $D[i,n]=1$ if and
  only if there are $p_0,\dotsc,p_i$ with $\sum_{0\leq j \leq
    i}p_j=n$. Note that the solution $S$ for $I$ can be obtained from
  the table $D$ as the set of all numbers $n$ such that $D[\ell,n]=1$. It
  hence remains to show how $D$ can be computed.

  We compute $D[i,n]$ via
  dynamic programming using the following recurrence relation.
  We start by setting $D[0,n]=1$ for every $n$ with $0 \leq n \leq
  \max(I)$ if and only if $n \in P_0$. Moreover, for every $i$
  with $1 \leq i \leq \ell$ and every $n$ with $0\leq n \leq
  \max(I)$, we set $D[i,n]=1$ if and only if
  there is an $n'$ with $n \leq n' \leq \max(I)$ and a $p \in
  P_i$ such that $n'+p=n$ and $D[i-1,n']=1$. 
  
  The running time of the algorithm is $\bigoh(\ell \cdot \max(I)^2)$
  since we require $\bigoh(\ell
  \cdot \max(I))$ to initialize the table $D$ and each of the
  $\ell$ recursive steps requires time $\bigoh(\max(I)^2)$.
\end{proof}

With Lemma~\ref{lem:solve-pss} in hand, we can proceed to a polynomial-time algorithm for \GRAPHCON{}.

\lv{\begin{LEM}}
\sv{\begin{LEM}[$\star$]}
\label{lem:graphcon-poly}
  An instance $I=(T,\lambda)$ of \GRAPHCON{} can be solved
  in time $\bigoh(|V(T)|^2\cdot \max(I)^2)$.
\end{LEM}
\begin{proof}
   As mentioned earlier, the main idea behind our 
   algorithm for \GRAPHCON{} is to use a dynamic programming algorithm
  working from the leaves to an arbitrarily chosen root $r$ of the tree
  $T$. Informally, the algorithm computes a set of numbers for each
  non-root vertex  $v$ of $T$
  representing the set of all assignments of the edge from $v$ to its
  parent that can be extended
  to a valid assignment of all edges in the subtree of $T$ rooted at
  $v$. Once this set has been computed for all children of the root we
  can construct a simple \PSS{} instance (given below) to decide whether $I$ has a
  solution.
  
  More formally, for a vertex $v$ of $T$ we denote by $T_v$ the subtree of $T$ rooted
  at $v$ and by $T_v^*$ the subtree of $T$ consisting of $T_v$ plus
  the edge between $v$ and its parent in $T$; for the root $r$ of $T$
  it holds that $T_v^*=T_v$.
  For every non-root vertex $v$ with parent $p$ we will
  compute a set $R(v)$ of numbers. Informally, $R(v)$ contains all
  numbers $n$ such that the assignment setting $\{p,v\}$ to $n$ can be
  extended to an assignment for all the edges in $T_v^*$ satisfying
  all constrains given by the vertices in $T_v$. More formally,
  $n \in R(v)$ if and only if there is an assignment $\alpha :
  E(T_v^*) \rightarrow \Nat$ with $\alpha(\{p,v\})=n$ such that
  for every $v \in V(T_v)$ it holds that
  $\sum_{e\in E(T_v^*) \land v \in e}\alpha(e)=\sum_{e\in E(T) \land v
    \in e}\alpha(e)\in \lambda(v)$.

  As stated above we will compute the sets $R(v)$ via a bottom-up
  dynamic programming algorithm starting at the leaves of $T$ and
  computing $R(v)$ for every inner node $v$ of $T$ using solely the
  computed sets $R(c)$ of each child $c$ of $v$ in $T$. Note
  that having computed $R(c)$ for every child $c$ of the root $r$ of
  $T$ we can decide whether $I$ has a solution as follows.
  Let $c_1,\dotsc, c_\ell$ be the children of $r$ in $T$; then $I$ has a
  solution if and only if the solution set for the instance
  $(\lambda(r),R(c_1),\dotsc,R(c_\ell))$ of \PSS{} contains $0$.

  It remains to show how to compute $R(v)$ for the leaves and inner
  nodes of $T$. If $v$ is a leaf then $R(v)$ is simply equal to
  $\lambda(v)$. Moreover, if $v$ is an inner node with children
  $c_1,\dotsc,c_\ell$, then $R(v)$ is equal to the solution set for the
  instance $(\lambda(v),R(c_1),\dotsc,R(c_\ell))$ of \PSS{}.
  This completes the description of the algorithm.

  The running time of the algorithm is at most $\bigoh(|V(T)|^2\cdot \max(I)^2)$
  since the time required at a leaf $q$ of $T$ is at most $\bigoh(\max(\lambda(q)))$ and
  the time required at any none-leaf node $t$ of $T$ with children
  $t_1,\dotsc, t_\ell$ is at most the time required to solve the instance
  $(\lambda(t),R(t_1),\dotsc,R(t_\ell))$ of \PSS{}, which is at most
  $\bigoh(|V(T)|\cdot \max(\lambda(v))^2)$ due to Lemma~\ref{lem:solve-pss}.
\end{proof}

\subsection{Multidimensional Partitioned Subset Sum}
\label{ssec:mpss}

Our second generalization of \SSP{} is a multi-dimensional variant of the problem that allows to separate the input set of
numbers into several groups, and restricts the solution to take at most
$1$ vector from each group. For technical reasons,
we will only search for solutions of size at most $r$.

\pbDef{\textsc{Multidimensional Partitioned Subset Sum} (\MDPSS)}{$k\in \Nat$, $r\in \Nat_0$, and a family $\PPP=\{P_1, \dotsc P_l\}$ of
  sets of vectors over $\Nat_0^k$.}
{Compute the set of all vectors $\bar{t} \in \{0,\dots,r\}^k$ such that
  there are $\bar{p}_1,\dotsc, \bar{p}_l$ with $\bar{p}_i \in P_i$ for
  every $i$ with $1 \leq i \leq l$ such that $\sum_{1\leq i \leq l}\bar{p}_i=\bar{t}$.}
  
It is easy to see that \textsc{Subset Sum} is a special case of
\MDPSS: given an instance of \textsc{Subset Sum}, we can create an
equivalent instance of \MDPSS\ by setting $r$ to a sufficiently large number and simply making each group $P_i$
contain two vectors: the all-zero vector and the vector that is equal
to the $i$-th number of the \SSP{} instance in all entries.

\lv{\begin{LEM}}
\sv{\begin{LEM}[$\star$]}
\label{lem:solve-mdpss}
  An instance $I=(k,r,\PPP)$ of \MDPSS{} can be solved in
  time $\bigoh(|\PPP| \cdot r^k)$. 
\end{LEM}
\begin{proof}
  We use a dynamic programming procedure similar to the approach
  used for the well-known \SSP{} problem~\cite{GareyJohnson79}. Let
  $I=(k,r,\PPP)$ with $\PPP=\{P_1,\dotsc,P_l\}$ be an instance of \MDPSS{}.

  We solve $I$ by 
  computing a table $D$ having one binary entry $D[i,\bar{t}]$ for every $i$
  and $\bar{t}$ with $0 \leq i
  \leq l$ and $\bar{t} \in [r]_0^k$ such that $D[i,\bar{t}]=1$ if and
  only if there are $\bar{p}_1,\dotsc, \bar{p}_i$ with $\bar{p}_j \in
  P_j$ for every $j$ with $1 \leq j \leq i$ such that
  $\sum_{1\leq j \leq i}\bar{p}_j=\bar{t}$.
  Note that the solution for $I$ can be obtained from the table
  $D$ as the set of all vectors $\bar{t} \in [r]_0^k$ such that
  $D[l,\bar{t}]=1$. 
  It hence remains to show how to compute the table $D$.

  We compute $D[i,\bar{t}]$ via
  dynamic programming using the following recurrence relation.
  We start by setting $D[1,\bar{t}]=1$ for every $\bar{t}\in
  [r]_0^k$ if and only if $\bar{t} \in P_1$.
  Moreover, for every $i$
  with $1 \leq i \leq l$ and every $\bar{t} \in [r]_0^k$,
  we set $D[i,\bar{t}]=1$ if and only if
  there is a $\bar{p}_j \in P_j$ with $\bar{p}_j\leq \bar{t}$
  such that $D[i-1,\bar{t}-\bar{p}_j]=1$. 
  
  The running time of the algorithm is $\bigoh(|\PPP| \cdot r^k)$
  since we require $\bigoh(|\PPP|
  \cdot r^k)$ to initialize the table $D$ and each of the
  $|\PPP|$ recursive steps requires time $\bigoh(r^k)$.
\end{proof}


\subsection{Simple Multidimensional Partitioned Subset Sum}
\label{ssec:smpss}
In this section, we are interested in a much more restrictive version
of \MDPSS{}, where all vectors (apart from the target vector) are only
allowed to have at most one non-zero component. Surprisingly, we
show that the \W{1}\hy hardness of the previously studied \textsc{Multidimensional Subset Sum} problem~\cite{GanianOrdyniakRamanujan17,GanianKO18} carries over to this more
restrictive variant using an intricate and involved reduction.

To formalize, we say that a set $P$ of vectors in $\Nat_0^d$ is \emph{simple}
if each vector in $P$ has exactly one non-zero component and the
values of the non-zero components for any two distinct vectors in $P$
are distinct.

\pbDefP{\textsc{Simple Multidimensional Partitioned Subset Sum} (\SMDPSS)}{$d\in \Nat$, $\bar{t}\in \Nat_0^d$, and a family $\PPP=\{P_1, \dotsc P_l\}$ of
  simple sets of vectors in $\Nat_0^d$.}{$d$.}
{Are there vectors $\bar{p}_1,\dotsc, \bar{p}_l$ with $\bar{p}_i \in P_i$ for
  every $i$ with $1 \leq i \leq l$ such that $\sum_{1\leq i \leq l}\bar{p}_i=\bar{t}$.}
\begin{THE}\label{the:smdpsshard}
  \SMDPSS{} is strongly \emph{\cc{W[1]}}\hy hard.
\end{THE}
\begin{proof}
  We will employ a
  parameterized reduction from the \mc{} problem, which is
  well-known to be \W{1}\hy complete~\cite{Pietrzak03}.

  \pbDefP{\mc{}}
  {An integer $k$, a $k$-partite graph $G=(V,E)$ with partition $\{V_1,\ldots,V_k\}$ of $V$ into sets of equal size.}
  {$k$}
  {Does $G$ have a $k$-clique, i.e., a set \mbox{$C\subseteq V$} of $k$ vertices such that $\forall u,v\in C$, with $u\neq v$ there is an edge $\{u,v\}\in E$?}
  
  We denote by $E_{i,j}$ the set of edges of $G$ that have one endpoint
  in $V_i$ and one endpoint in $V_j$ and we assume w.l.o.g. that
  $|V_i=\{v_1^i,\dotsc,v_n^i\}|=n$ and $|E_{i,j}|=m$ for every $i$ and $j$ with $1 \leq i < j
  \leq k$ (see the standard parameterized complexity textbook
  for a justification of these assumptions~\cite{CyganFKLMPPS15}). 
    \newcommand{\polyZ}{n^{2}}
  \newcommand{\polyF}{n^{8}}
  \newcommand{\polyS}{n^6}
  \newcommand{\polyT}{n^4}
  \newcommand{\indJ}[2]{\textsf{indJ}(#1,#2)}
  \newcommand{\indM}[1]{\textsf{indMin}(#1)}
  \newcommand{\indMa}[1]{\textsf{indMax}(#1)}

Given an instance $(G,k)$ of \mc{} with partition $V_1,\dotsc,V_k$,
  we construct an equivalent instance $I=(d,\bar{t},\PPP)$ of \SMDPSS{} in
  polynomial time, where $d=k(k-1)+\binom{k}{2}$ and $|\PPP|=\binom{k}{2}+nk(2k-3)$.
  We will also make use of the following notation.
  For $i$ and $j$ with $1 \leq i \leq k$ and $1 \leq j < k$, we denote
  by $\indJ{i}{j}$ the $j$-th smallest number in $[k]\setminus \{i\}$ and we
  denote by $\indM{i}$ and $\indMa{i}$ the numbers $\indJ{i}{1}$ and
  $\indJ{i}{k-1}$, respectively.

  We assign to every vertex $v$ of $G$ a unique
  number $\SSS(v)$ from a Sidon sequence $\SSS$ of length
  $|V(G)|$~\cite{ErdosTuran41}. A \emph{Sidon sequence} is a
  sequence of natural numbers such that the sum of each pair of numbers is
  unique; it can be shown that it is possible to construct
  such sequences whose maximum value is bounded by a polynomial in its
  length~\cite{AignerZiegler10,ErdosTuran41}. 

  To simplify the description of $I$, we will introduce names and
  notions to identify both components of vectors and sets in $\PPP$.
  Every vector in $I$ has the following components:
  \begin{itemize}
  \item For every $i$ and $j$ with $1 \leq i,j \leq k$ and $i\neq j$,
    the \emph{vertex component} $c_V^i(j)$. We set $\bar{t}[c_V^i(j)]$ to:
    \begin{itemize}
    \item $\polyS+\polyT$ if $j=\indM{i}$,
    \item $(n-1)\polyF+\polyS+\polyT+\sum_{\ell=1}^n(\ell+\ell \polyZ)$ if $j> \indM{i}$ and $j<\indMa{i}$, and
    \item $(n-1)\polyF+\polyS+\sum_{\ell=1}^n(\ell)$, otherwise.
    \end{itemize}

  \item For every $i$ and $j$ with $1 \leq i < j \leq k$, the
    \emph{edge component} $c_E(i,j)$
    with $\bar{t}[c_E(i,j)]=\sum_{v \in V_i \cup V_j}\SSS(v)$.

  \end{itemize}
  Note that the total number of components $d$ is equal to
  $k(k-1)+\binom{k}{2}$ and that for every $i$ with $1 \leq i \leq k$, there are $k-1$
  vertex components, i.e., the components
  $c_V^i(\indJ{i}{1}),\dotsc,c_V^i(\indJ{i}{k-1})$, which intuitively
  have the following tasks.
  The first component, i.e., the component $c_V^i(\indJ{i}{1})$
  identifies a vertex $v \in V_i$ that should be part of a
  $k$-clique in $G$. Moreover, every component $c_V^i(\indJ{i}{j})$ (including
  the first component), is also responsible for: (1) Signalling the
  choice of the chosen vertex $v \in V_i$ to the next component, i.e.,
  the component $c_V^i(\indJ{i}{j+1})$ and (2) Signalling the
  choice of the vertex $v \in V_i$ to the component $c_E(i,j)$ 
  that will then verify that there is an edge between the vertex
  chosen for $V_i$ and the vertex chosen for $V_j$.
  This interplay between the components will be achieved through the
  sets of vectors in $\PPP$ that will be defined and explained next.
  
  \begin{table}[ht]
    \centering
    \begin{tabular}{r|ccccc|c}
      \toprule
      & $P_{EV}^1(2,\ell)$ & $P_V^1(2,\ell)$ & $P_{EV}^1(3,\ell)$ & $P_V^1(3,\ell)$ & $P_{EV}^1(4,\ell)$ & $\bar{t}$ \\
      \midrule
      $c_V^1(2)$ & $\polyS +\ell$ & $\polyT-\ell$ & & & & $\polyS+\polyT$ \\
      $c_V^1(3)$ & & $\polyF+\ell+\ell \polyZ$ & $\polyS+\ell$ &
                                                              $\polyT+\ell
                                                              \polyZ$ & &
      $Z+\polyT+\sum_{\ell=1}^n(\ell \polyZ)$\\
      $c_V^1(4)$ & & & & $\polyF+\ell$ & $\polyS+\ell$ &
                                                         $Z$\\
      $c_E(1,2)$ & $\SSS(v_\ell^1)$ & & & & &  $\sum_{v \in V_1 \cup
                                                V_2}\SSS(v)$ \\
      $c_E(1,3)$ & & & $\SSS(v_\ell^1)$ & & & $\sum_{v \in V_1 \cup
                                               V_3}\SSS(v)$ \\
      $c_E(1,4)$ & && & & $\SSS(v_\ell^1)$ & $\sum_{v \in V_1 \cup V_4}\SSS(v)$ \\
    \end{tabular}
    \caption{An illustration of the vectors contained in the sets
      $P_{EV}^1(2,\ell),\dotsc,P_{EV}^1(4,\ell)$ and 
      $P_V^1(2,\ell),\dotsc,P_V^1(3,\ell)$. For example the column for
      the set $P_{EV}^1(2,\ell)$ shows that the set contains two
      vectors, one whose only non-zero component $c_V^1(2)$ 
      has the value $\polyS+\ell$ and a second one whose only non-zero
      component $c_E(1,2)$ and has the value $\SSS(v_\ell^1)$. The last
      column provides the value for the target vector for the
      component given by the row. Finally, the value $Z$ is equal to
      $(n-1)\polyF+\polyS+\sum_{l=1}^n(\ell)$.}
    \label{tbl:sgasp-hard-1}
  \end{table}

  \begin{table}[ht]
    \centering
    \begin{tabular}{r|ccc|c}
      \toprule
      & $P_{EV}^i(j,\ell)$ & $P_{EV}^j(i,\ell)$ & $P_{E}(i,j)$ & $\bar{t}$ \\
      \midrule
      $c_V^i(j)$ & $\polyS +\ell$ & & & \\
      $c_V^j(i)$ & & $\polyS +\ell$ & &\\
                                        
      $c_E(i,j)$ & $\SSS(v_\ell^i)$ & $\SSS(v_\ell^j)$ & $\{
                                                         \SSS(v)+\SSS(u)
                                                               \SM
                                                         \{v,u\} \in
                                                         E_{i,j}\SE$ & $\sum_{v \in V_i \cup V_j}\SSS(v)$
    \end{tabular}
    \caption{An illustration of the vectors contained in the sets
      $P_{EV}^i(j,\ell)$, $P_{EV}^j(i,\ell)$, and $P_E(i,j)$ and their
      interplay with the components $c_V^i(j)$, $c_V^j(i)$, and
      $c_V(i,j)$. For the conventions used in the table please refer
      to Table~\ref{tbl:sgasp-hard-1}. Additionally, note that the
      column for $P_E(i,j)$ indicates that the set contains one vector
      for every edge $\{v,u\}\in E_{i,j}$, whose only non-zero
      component $c_E(i,j)$ has the value $\SSS(v)+\SSS(u)$.}
    \label{tbl:sgasp-hard-2}
  \end{table}  
  
  $\PPP$ consists of the following sets, which are illustrated in Table~\ref{tbl:sgasp-hard-1} and~\ref{tbl:sgasp-hard-2}:
  \begin{itemize}
  \item For every $i$, $j'$, and $\ell$ with $1 \leq i \leq k$, $1\leq j' \leq k-2$,
    and $1 \leq \ell \leq n$, the \emph{vertex set}
    $P_V^i(j,\ell)$, where $j=\indJ{i}{j'}$, containing two vectors $\bar{v}_{i,j,\ell}^+$ and
    $\bar{v}_{i,j,\ell}^-$ defined as follows:
    \begin{itemize}
    \item if $j'=1$, then
      $\bar{v}_{i,j,\ell}^+[c_V^i(j)]=\polyT-\ell$ and
      $\bar{v}_{i,j,\ell}^-[c_V^i(\indJ{i}{j'+1})]=\polyF+\ell+\ell \polyZ$ or
    \item if $1 < j'< k-2$, then
      $\bar{v}_{i,j,\ell}^+[c_V^i(j)]=\polyT+\ell \polyZ$ and
      $\bar{v}_{i,j,\ell}^-[c_V^i(\indJ{i}{j'+1})]=\polyF+\ell+\ell \polyZ$ or
    \item if $j'=k-2$, then
      $\bar{v}_{i,j,\ell}^+[c_V^i(j)]=\polyT+\ell \polyZ$ and $\bar{v}_{i,j,\ell}^-[c_V^i(\indJ{i}{j'+1})]=\polyF+\ell$.
    \end{itemize}

    We denote by $P_V^{i}(j)$, $P_{V+}^{i}(j)$, and $P_{V-}^{i}(j)$ the
    sets $\bigcup_{\ell=1}^n(P_V^i(j,\ell))$, $P_V^i(j)\cap \SB
    \bar{v}_{i,j,\ell}^+\SM 1 \leq \ell \leq n\SE$, and
    $P_V^i(j)\setminus P_{V+}^{i}(j)$, respectively.
    
  \item For every $i$, $j$, and $\ell$ with $1\leq i,j \leq k$, $i\neq
    j$, and $1 \leq \ell \leq n$, the \emph{vertex incidence set}
    $P_{EV}^i(j,\ell)$, which contains the two vectors
    $\bar{a}_{i,j,\ell}^+$ and $\bar{a}_{i,j,\ell}^-$
    such that $\bar{a}_{i,j,\ell}^+[c_V^i(j)]=\polyS+\ell$ and
    $\bar{a}_{i,j,\ell}^-[c_E(i,j)]=\SSS(v_\ell^i)$.

    We denote by $P_{EV}^{i}(j)$, $P_{EV+}^{i}(j)$, and $P_{EV-}^{i}(j)$ the
    sets $\bigcup_{\ell=1}^n(P_{EV}^i(j,\ell))$, $P_V^i(j)\cap \SB
    \bar{a}_{i,j,\ell}^+\SM 1 \leq \ell \leq n\SE$, and
    $P_{EV}^i(j)\setminus P_{EV+}^{i}(j)$, respectively.

  \item For every $i$, $j$ with $1\leq i < j\leq k$, the \emph{edge
      set} $P_{E}(i,j)$, which for every $e=\{v,u\} \in E_{i,j}$
    contains the vector $\bar{e}$ such that
    $\bar{e}[c_E(i,j)]=\SSS(v)+\SSS(u)$; note that $P_E(i,j)$ is
    indeed a simple set, because $\SSS$ is a Sidon sequence.

  \end{itemize}
  Note that altogether there are
  $nk(k-2)+\binom{k}{2}+nk(k-1)=\binom{k}{2}+nk(2k-3)$ sets in $\PPP$.

  Informally, the two vectors $\bar{v}_{i,j,\ell}^+$ and $\bar{v}_{i,j,\ell}^-$ in $P_V^i(j,\ell)$
  represent the choice of whether or not the vertex $v_\ell^i$ should
  be included in a $k$-clique for $G$, i.e., if a solution for $I$
  chooses $v_{i,j,\ell}^+$ then $v_\ell^i$ should be part of a
  $k$-clique and otherwise not. The component $c_V^i(j)$, more
  specifically the value for $\bar{t}[c_V^i(j)]$, now ensures that a
  solution can choose at most one such vector in
  $P_{V+}^i(j)$. Moreover, the fact that all but one of the vectors
  $\bar{v}_{i,j,1}^-,\dotsc,\bar{v}_{i,j,n}^-$ need to be chosen by a
  solution for $I$ signals the choice of the vertex for $V_i$ to the
  next component, i.e., either the component $c_V^i(j+1)$ if $j+1\neq
  i$ or the component $c_V^i(j+2)$ if $j+1=i$. Note that we only need $k-2$
  sets $P_V^i(j)$ for every $i$, because we need to copy the
  vertex choice for $V_i$ to only $k-1$ components.
  A similar idea
  underlies the two vectors $\bar{a}_{i,j,\ell}^+$ and
  $\bar{a}_{i,j,\ell}^-$ in $P_{EV}^i(j,\ell)$, i.e., again the
  component $c_V^i(j)$ ensures that $\bar{a}_{i,j,\ell}^+$ can be
  chosen for only one of the sets $P_{EV}^i(j,1),\dotsc,P_{EV}^i(j,n)$
  and $\bar{a}_{i,j,\ell}^-$ must be chosen for all the remaining
  ones. Note that the component $c_V^i(j)$ now also ensures that the
  choice made for the sets in $P_V^i(j)$ is the same as the choice made
  for the sets in $P_{EV}^i(j)$. Moreover, the choice made for the
  sets in $P_{EV}^i(j)$ is now propagated to the component $c_E(i,j)$
  (instead of the next vertex component). Finally, the vectors in the
  set $P_E(i,j)$ represent the choice of the edge used in a $k$-clique
  between $V_i$ and $V_j$ and the component $c_E(i,j)$ ensures that
  only an edge, whose endpoints are the two vertices signalled by the
  sets $P_{EV}^i(j)$ and $P_{EV}^j(i)$ can be chosen. We are now ready
  to give a formal proof for the equivalence between the two instances.

  This completes the construction of $I$. It is straightforward to
  verify that all sets in $\PPP$ are simple, $I$ can be constructed in
  polynomial-time, and all the component values of all vectors are
  bounded by a polynomial in $n$. It hence only remains to show that
  $(G,k)$ has a solution if and only if so does $I$. 
  
  Towards showing the forward direction, let $c_1,\dotsc,c_k$ with
  $c_i \in V_i$ be the
  vertices of a $k$-clique of $G$ and for every $i$ and $j$ with $1
  \leq i < j \leq k$, let $e_{i,j}$ be the edge between $c_i$ and
  $c_j$ in $G$. We obtain a solution $S \subseteq \PPP$ for $I$ with $\sum_{\bar{s}
    \in S}\bar{s}=\bar{t}$ and $|S \cap P|=1$ for every $P \in \PPP$
  by choosing the following vectors:
  \begin{itemize}
  \item For every $i$ and $j$ with $1 \leq i,j \leq k$ and $j\neq i$,
    we choose the vector $\bar{a}_{i,j,\ell}^+$ from the set
    $P_{EV}^i(j,\ell)$ if $v_\ell^i=c_i$ and otherwise we choose the
    vector $\bar{a}_{i,j,\ell}^-$.
  \item For every $i$ and $j$ with $1 \leq i,j \leq k$ and $j\notin\{i,\indMa{i}\}$,
    we choose the vector $\bar{v}_{i,j,\ell}^+$ from the set
    $P_{V}^i(j,\ell)$ if $v_\ell^i=c_i$ and otherwise we choose the
    vector $\bar{v}_{i,j,\ell}^-$.
  \item For every $i$ and $j$ with $1 \leq i < j \leq k$, we choose
    the vector $\bar{e_{i,j}}$ from the set $P_E(i,j)$.
  \end{itemize}
  It is straightforward to verify that the above choices constitute a
  solution for $I$.
  
  Towards showing the reverse direction, let $S \subseteq \bigcup_{P\in\PPP}P$ be a solution for $I$, i.e.,
  $\sum_{\bar{s} \in S}\bar{s}=\bar{t}$ and $|S\cap P|=1$
  for every $P \in \PPP$. We show the reverse direction using the
  following series of claims.

  \begin{itemize}
  \item[(C1)] For every $i$ and $j$ with $1\leq i,j\leq k$ and $i\neq
    j$, it holds that $|S\cap P_{EV+}^i(j)|=1$ and $|S\cap
    P_{EV-}^i(j)|=n-1$. In the following let $\bar{v}_{EV}^i(j)$ be the
    unique vector in $S\cap P_{EV+}^i(j)$.
  \item[(C2)] For every $i$ and $j'$ with $1 \leq i\leq k$ and $1\leq
    j'\leq k-2$, it holds that $|S\cap P_{V+}^i(j)|=1$
    and $|S\cap P_{V-}^i(j)|=n-1$, where $j=\indJ{i}{j'}$. In the following let $\bar{v}_V^i(j)$ be the
    unique vector in $S\cap P_{V+}^i(j)$.
  \item[(C3)] For every $i$, $j$ and $j'$ with $1\leq i,j,j' \leq k$,
    $j\neq i$, $j'\neq i$, and $j'\neq j$, it holds that
    $\bar{v}_{EV}^i(j)[c_V^i(j)]=\bar{v}_{EV}^i(j')[c_V^i(j')]$. In other words
    for $i$ as above, there exists a unique value $\ell_i$ such that
    $\bar{v}_{EV}^i(j)[c_V^i(j)]=\polyS+\ell_i$ for every $j$ as
    above.
  \item[(C4)] For every $i$ and $j$ with $1 \leq i < j \leq k$, $G$
    contains an edge between $v_{\ell_i}^i$ and $v_{\ell_j}^j$.
  \item[(C5)] The vertices $v_{\ell_1}^1,\dotsc,v_{\ell_k}^k$ induce a
    clique in $G$. 
  \end{itemize}

  Towards showing (C1) consider the component $c_V^i(j)$. Then
  $\bar{t}[c_v^i(j)]$ contains the term $\polyS$ and moreover the only vectors
  of $I$ having a non-zero component at $c_V^i(j)$ (apart from the
  vectors in $P_{EV+}^i(j)$) are the vectors in $P_{V-}^i(j')$ (only
  if $j>\indM{i}$), where $j'=j-1$ if $j-1\neq i$ and $j'=j-2$
  otherwise,
  and the vectors in $P_{V+}^i(j')$ (only if
  $j<\indMa{i}$). The former all have values larger than $\polyF\geq \polyS$ and
  the sum of all values of the latter is at most
  $\sum_{\ell=1}^n\polyT+\ell\polyZ\leq 2n^5< \polyS$. Hence the
  only vectors that can contribute the term $\polyS$ are the vectors
  in $P_{EV+}^i(j)$ and since $\polyS$ appears exactly once in
  $\bar{t}[c_V^i(j)]$, (C1) follows. 

  Towards showing (C2) consider the component $c_V^i(j)$. Then
  $\bar{t}[c_V^i(j)]$ contains the term $\polyT$ and moreover the only vectors
  of $I$ having a non-zero component at $c_V^i(j)$ (apart from the
  vectors in $P_{E+}^i(j)$) are the vectors in $P_{V-}^i(j')$ (only
  if $j>\indM{i}$), where $j'=j-1$ if $j-1\neq i$ and $j'=j-2$
  otherwise, and the vectors in $P_{EV-}^i(j)$. The former and
  the latter have values larger than $\polyF$ and $\polyS$,
  respectively. Hence the
  only vectors that can contribute the term $\polyT$ are the vectors
  in $P_{V+}^i(j)$ and since $\polyT$ appears exactly once in
  $\bar{t}[c_V^i(j)]$, (C2) follows. 

  Towards showing (C3), we show that
  $\bar{v}_{EV}^i(j)[c_V^i(j)]=\bar{v}_{EV}^i(j')[c_V^i(j')]$, where 
  $j=\indJ{i}{r}$ and $j'=\indJ{i}{r+1}$ for every $r$ with $1 \leq r
  <k$. Since we can assume that w.l.o.g. $k>3$, we only need to
  distinguish the following three cases:
  \begin{itemize}
  \item[(A)] $r=1$ and $r+1<k-1$,
  \item[(B)] $r>1$ and $r+1<k-1$,
  \item[(C)] $r>1$ and $r+1=k-1$.
  \end{itemize}
  For the case (A), consider the component $c_V^i(j)$. Note that due to (C1) and
  (C2), the vectors $\bar{v}_{EV}^i(j)$ and $\bar{v}_V^i(j)$ are the only
  vectors in $S$, for which the component $c_V^i(j)$ is non-zero. Hence,
  $\bar{v}_{EV}^i(j)[c_V^i(j)]+\bar{v}_V^i(j)[c_V^i(j)]=\bar{t}[c_V^i(j)]=\polyS+\polyT$,
  which is only possible if $\bar{v}_{EV}^i(j)=\bar{a}_{i,j,\ell_1}^+$ and
  $\bar{v}_V^i(j)=\bar{v}_{i,j,\ell_1}^+$ for some $\ell_1$ with $1 \leq \ell_1
  \leq n$. Now consider the component $c_V^i(j')$. Because of (C2), we obtain that $\sum_{s \in S \cap
    P_{V-}^i(j)}s[c_V^i(j)]=(\sum_{\ell=1}^n\polyF+\ell+\ell\polyZ)-(\polyF+\ell_1+\ell_1
  \polyZ)$. Moreover, because of (C1) and (C2), we obtain that
  $(\sum_{\ell=1}^n\polyF+\ell+\ell\polyZ)-(\polyF+\ell_1+\ell_1\polyZ)+\bar{v}_{EV}^i(j')+\bar{v}_V^i(j')=\bar{t}[c_V^i(j')$, which is only possible if $\bar{v}_{EV}^i(j')=\bar{a}_{i,j',\ell_1}^+$ and
  $\bar{v}_V^i(j')=\bar{v}_{i,j',\ell_1}^+$. Hence
  $\bar{v}_{EV}^i(j)[c_V^i(j)]=\bar{v}_{EV}^i(j')[c_V^i(j')]$, as required.
  The proof for the cases (B) and (C) is analogous. 

  Towards showing (C4) consider the component $c_E(i,j)$. Note that
  the set $P_{EV}^i(j)$, $P_{EV}^j(i)$, and $P_E(i,j)$ are the only
  sets in $\PPP$ containing vectors that are non-zero at
  $c_E(i,j)$. Moreover, because of (C1) it holds that $\sum_{\bar{s}
    \in S\cap
    P_{EV}^i(j)}\bar{s}[c_E(i,j)]=(\sum_{\ell=1}^n\SSS(v_\ell^i))-\SSS(v_{\ell_i}^i)$
  and similarly $\sum_{\bar{s} \in S\cap
    P_{EV}^j(i)}\bar{s}[c_E(i,j)]=(\sum_{\ell=1}^n\SSS(v_\ell^j))-\SSS(v_{\ell_j}^j)$. Since
  $\bar{t}[c_E(i,j)]=\sum_{v\in V_i\cup V_j}\SSS(v)$, we obtain
  that the unique vector $\bar{e} \in S \cap P_E(i,j)$ must satisfy
  $\bar{e}[c_E(i,j)]=\SSS(v_{\ell_i}^i)+\SSS(v_{\ell_j}^j)$, which due
  the properties of Sidon sequences is only possible if $e$ is an edge
  between $v_{\ell_i}^i$ and $v_{\ell_j}^j$ in $G$. Finally, (C5)
  follows immediately from (C3) and (C4).
\end{proof}

\section{Result 1: Fixed-Parameter Tractability of \sGasp{}}
\label{sec:res1}

In this section we will establish that \sGasp\ is \FPT\ when parameterized
by the number of agent types and the number of activities by proving
Theorem~\ref{the:fpt-agents-act}.

\begin{THE}\label{the:fpt-agents-act}
  \sGasp{} can be solved in time
$\bigoh(2^{|T(N)|\cdot(1+|A|)}\cdot ((|N|+|A|)|N|)^{2})$.
\end{THE}

Let $I=(N,A,(P_n)_{n \in N})$ be a \sGasp\ instance and let $\pi
: N \rightarrow A^*$ be an assignment of agents to activities. We
denote by $G_I(\pi)$ the incidence graph between $T(N)$ and $A$, which
is defined as follows. $G_I(\pi)$ has vertices $T(N) \cup A$ and
contains an edge between an agent type $t \in T(N)$ and an activity $a
\in A$ if $\pi_{t,a}\neq \emptyset$. We say that $\pi$ is \emph{acyclic} if
$G_I(\pi)$ is acyclic.

Our first aim towards the proof of
Theorem~\ref{the:fpt-agents-act} is to show that if $I$ has a stable assignment, then
it also has an acyclic stable assignment
(Lemma~\ref{lem:fpt-acyc}). We will then show in
Lemma~\ref{lem:fpt-acyclic-poly} that finding a stable assignment whose
incidence graph is equal to some given acyclic pattern graph
can be achieved in
polynomial-time via a reduction to the \GRAPHCON{} problem defined and solved in Subsection~\ref{ssec:tss}.
Since the number of (acyclic) pattern graphs
is bounded in our parameters, we can subsequently solve \sGasp{} by enumerating
all acyclic pattern graphs and checking for each of them whether there is an
acyclic solution matching the selected pattern.

A crucial notion towards showing that it is sufficient to consider
only acyclic solutions is the notion of (strict) compression.
We say that an assignment $\tau$ is a \emph{compression} of $\pi$ if
it satisfies the following conditions:
\begin{itemize}
\item[(C1)] for every $t \in T(N)$ it holds that
  $|\pi_{t}|=|\tau_{t}|$, 
\item[(C2)] for every $a \in A$ it holds that
  $|\pi^{-1}(a)|=|\tau^{-1}(a)|$, 
  and
\item[(C3)] for every $a \in A$ it holds that the set of agent types $\tau$ assigns to $a$ is a subset of the agent types $\pi$ assigns to $a$.
\end{itemize}
Moreover, if $\tau$ additionally satisfies:
\begin{itemize}
\item[(C4)] $|E(G_I(\tau))|< |E(G_I(\pi))|$,
\end{itemize}
 we say that $\tau$ is a \emph{strict compression} of $\pi$.

Intuitively, an assignment $\tau$ is a compression of $\pi$ if it
maintains all the properties required to preserve stability and
compatibility with a given subset $Q \subseteq T(N)$.
We note that condition (C3) can be formalized as $T(\tau^{-1}(a)) \subseteq
  T(\pi^{-1}(a))$.
Observe that if $\tau$ is a strict compression then there is at least one
activity $a \in A$ such that $T(\tau^{-1}(a)) \subset
T(\pi^{-1}(a))$.
The following lemma shows that every assignment that is not acyclic admits a strict compression.
\lv{\begin{LEM}}
\sv{\begin{LEM}[$\star$]}
\label{lem:fpt-compress}
  If $\pi$ is not acyclic, then there exists an assignment $\tau$
  that strictly compresses~$\pi$.
\end{LEM}
\begin{proof}
  \begin{figure}[tb]
\begin{tikzpicture}
      \tikzstyle{every node}=[]
      \tikzstyle{gn}=[circle, inner sep=4pt,draw, node distance=2cm]
      \tikzstyle{every edge}=[draw, line width=2pt]
      \draw
      node[gn, label=below:$t_1$] (t1) {}
      node[gn, below of=t1, label=below:$t_2$] (t2) {}
      node[gn, below of=t2, label=below:$t_3$] (t3) {}

      node[node distance=4cm,right of=t1] (d1) {}
      node[gn, node distance=1cm, below of=d1, label=below:$a_1$] (a1) {}
      node[gn, below of=a1, label=below:$a_2$] (a2) {}
      node[gn, below of=a2, label=below:$a_3$] (a3) {}
      ;
      
      \draw[line width=2pt]
      (t1) -- (a1) node[near end, label=above:$-1$] {}
      (a1) -- (t2) node[near end, label=above:$+1$] {}
      (t2) -- (a2) node[near end, label=above:$-1$] {}
      (a2) -- (t3) node[near end, label=above:$+1$] {}
      (t3) -- (a3) node[near end, label=above:$-1$] {}
      (a3) to[out=200, in=300] ($(t3)-(0.5cm,1cm)$) to[out=120,
      in=270] ($(t2)-(1cm,1cm)$) node[label=right:$+1$] {} 
      to[out=90, in=225] (t1);
    \end{tikzpicture}
    \caption{Illustration of the modification (M) in the
      proof of Lemma~\ref{lem:fpt-compress} for a cycle of length
      six. A label of $+1$ on an edge $\{t,a\}$ of $G_I(\pi)$ means
      that $|\kappa_{t,a}|$ is by one larger than
      $|\pi_{t,a}|$. Similarly,
      a label of $-1$ on an edge $\{t,a\}$ of $G_I(\pi)$ means
      that $|\kappa_{t,a}|$ is by one smaller than $|\pi_{t,a}|$. }
      \label{fig:acyclic-mods}
  \end{figure}
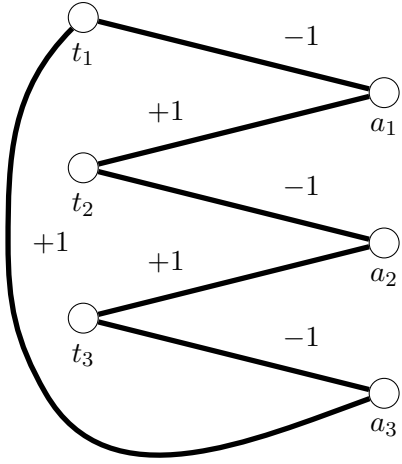

  Let $C=(t_1,a_1,\dotsc,t_l,a_l,t_1)$ be a cycle of $G_I(\pi)$. Consider
  the following modification (M) of our instance:  
  reassign one (arbitrary) agent in $\pi_{t_1,a_1}$ to
  $a_l$, and for every $i$ with $1 < i \leq l$ reassign
  one (arbitrary) agent in $\pi_{t_i,a_i}$ to $a_{i-1}$.
  An illustration of (M) is also provided in Figure~\ref{fig:acyclic-mods}.
  
  First, we show that the assignment $\kappa$ obtained from $\pi$ after
  applying modification (M) satisfies (C1)--(C3).
  Towards showing (C1),
  observe that $|\pi_{t}|=|\kappa_{t}|$ for any $t \in T(N) \setminus
  \{t_1,\dotsc,t_l\}$. Moreover, for every $i$ with $1 < i \leq l$,
  we have:
  \begin{eqnarray*}
    |\kappa_{t_{i}}| & = & \sum_{a \in A}|\kappa_{t_i,a}|\\
    & = & \big(\sum_{a \in (A \setminus
          \{a_{i},a_{i-1}\})}|\kappa_{t_i,a}|)+|\kappa_{t_i,a_i}|+ |\kappa_{t_i,a_{i-1}}|\big)\\
                   & = & \big(\sum_{a \in (A \setminus \{a_{i},a_{i-1}\})}|\pi_{t_i,a}|)+ (|\pi_{t_i,a_i}|-1)+(|\pi_{t_i,a_{i-1}}|+1\big)\\
    & = & \sum_{a \in A}|\pi_{t_i,a}| \\
    & = & |\pi_{t_{i}}|
  \end{eqnarray*}
  and for $i=1$ we have:
  \begin{eqnarray*}
    |\kappa_{t_{1}}| & = & \sum_{a \in A}|\kappa_{t_1,a}|\\
    & = & \big(\sum_{a \in (A \setminus \{a_{1},a_{l}\})}|\kappa_{t_i,a})|+|\kappa_{t_1,a_1}|+|\kappa_{t_1,a_{l}}|\big) \\
                   & = & \big(\sum_{a \in (A \setminus \{a_{1},a_{l}\})}|\pi_{t_1,a}|\big) +(|\pi_{t_1,a_1}|-1)+(|\pi_{t_1,a_{l}}|+1)\\
    & = & \sum_{a \in A}|\pi_{t_1,a}| \\
    & = & |\pi_{t_{1}}|
  \end{eqnarray*}
  Towards showing
  (C2) note that $|\pi^{-1}(a)|=|\kappa^{-1}(a)|$ for every $a \in A
  \setminus \{a_1,\dotsc,a_l\}$ and moreover for every $i$ with $1
  \leq i < l$ we obtain:
  \begin{eqnarray*}
    |\kappa^{-1}(a_{i})| & = & \sum_{t \in T(N)}|\kappa_{t,a_i}|\\
    & = & \big(\sum_{t \in (T(N) \setminus
          \{t_{i},t_{i+1}\})}|\kappa_{t,a_i})| +|\kappa_{t_i,a_i}|+|\kappa_{t_{i+1},a_{i}}|\\
    & = & (\sum_{t \in (T(N) \setminus\{t_{i},t_{i+1}\})}|\pi_{t,a_i})|+(|\pi_{t_i,a_i}|-1)+(|\pi_{t_{i+1},a_{i}}|+1)\\
    & = & \sum_{t \in T(N)}|\pi_{t,a_i}| \\
    & = & |\pi^{-1}(a_i)|
  \end{eqnarray*}
  and for $i=l$ we have:
  \begin{eqnarray*}
    |\kappa^{-1}(a_{l})| & = & \sum_{t \in T(N)}|\kappa_{t,a_l}|\\
    & = & (\sum_{t \in (T(N) \setminus \{t_{l},t_{1}\})}|\kappa_{t,a_l})|+|\kappa_{t_l,a_l}|+|\kappa_{t_{1},a_{l}}| \\
                        & = & (\sum_{t \in (T(N) \setminus \{t_{l},t_{1}\})}|\pi_{t,a_l})|+(|\pi_{t_l,a_l}|-1)+(|\pi_{t_{1},a_{l}}|+1) \\
    & = & \sum_{t \in T(N)}|\pi_{t,a_l}| \\
    & = & |\pi^{-1}(a_l)|
  \end{eqnarray*}

  Because $|\kappa_{t,a}|$ and $|\pi_{t,a}|$ can only differ if
  $|\pi_{t,a}|\neq 0$, i.e., $|\kappa_{t,a}|$ will never be non-equal to zero
  if $|\pi_{t,a}|=0$, we obtain that
  $|\kappa|$ satisfies also condition (C3).
  
  Having settled that (M) does not violate (C1)--(C3), we observe
  that if $G_I(\kappa)$ still contains the cycle $C$, we can apply
  modification (M) again to $\kappa$ and the obtained assignment still
  satisfies conditions (C1)--(C3). Hence we can repeatedly apply
  modification (M) as long as the cycle $C$ is not destroyed in the
  resulting assignment. Namely let
  $m=\min\{\pi_{t_1,a_1},\dotsc,\pi_{t_l,a_l}\}$, let $i$ be an
  index with $m=\pi_{t_i,a_i}$, and let $\tau$ be the
  assignment obtained after $m$ applications of modification (M) to
  $\pi$. Note that $m$ applications of modification (M) are possible
  since the cycle $C$ remains preserved up to the $m-1$-th modification of
  (M). Furthermore, since $|\tau_{t_i,a_i}|=0$, we obtain that $\tau$
  satisfies (C4).
\end{proof}

The following lemma shows that any assignment can be compressed into
an acyclic assignment. 
\sv{\begin{LEM}[$\star$]}
  \lv{\begin{LEM}}\label{lem:fpt-acyclic-compress}
  Let $\pi : N \rightarrow A^*$ be an assignment for $I$. Then there
  exists an acyclic assignment $\pi'$ that compresses $\pi$.
\end{LEM}
\begin{proof}
  The lemma follows via an exhaustive application of
  Lemma~\ref{lem:fpt-compress}. Namely, we start by checking whether
  $\pi$ is acyclic. If yes, then $\pi$ itself is the acyclic assignment that
  compresses $\pi$. If not, then we apply Lemma~\ref{lem:fpt-compress} to
  $\pi$ and obtain the assignment $\pi'$ that strictly compresses
  $\pi$. If $\pi'$ is acyclic, we are done; otherwise, we 
  repeat the above procedure with $\pi'$ instead of
  $\pi$. Because at every step $|E(G_I(\pi))|<|E(G_I(\pi'))|$ and
  $G_I(\pi)$ has at most $|T(N)|\cdot |A|$ edges, this process
  concludes after at most $|T(N)|\cdot |A|$ steps and results in an acyclic
  assignment that compresses~$\pi$.
\end{proof}

The following lemma provides the first cornerstone for our algorithm
by showing that it is sufficient to consider only acyclic solutions.
Intuitively, it is a consequence of Lemma~\ref{lem:fpt-acyclic-compress} 
along with the observation
that compression preserves stability and individual rationality.
\sv{\begin{LEM}[$\star$]}
  \lv{\begin{LEM}}\label{lem:fpt-acyc}
  If $I$ has a stable assignment, then $I$ has an acyclic stable assignment.
\end{LEM}
\begin{proof}
  Let $\pi$ be a stable assignment for $I$. Because of
  Lemma~\ref{lem:fpt-acyclic-compress}, there is an acyclic assignment
  $\pi'$ that compresses $\pi$. We claim that $\pi'$ is also a stable
  assignment. Conditions (C2) and (C3) together with the fact that
  $\pi$ is individually rational imply that also $\pi'$ is individually
  rational. Moreover, it follows from Condition (C1) that
  $\PE(I,\pi)=\PE(I,\pi')$, which together with Condition (C2) and the
  stability of $\pi$ implies the stability of $\pi'$.
\end{proof}

Our next step is the introduction of terminology related to the pattern graphs
mentioned at the beginning of this section.
Let $G$ be a bipartite graph with bi-partition $\{T(N),A\}$. We say
that $G$ \emph{models} an assignment $\pi : N \rightarrow A^*$ if
$G_I(\pi)=G$; in this sense every such bipartite graph can be seen as
a pattern (or model) for assignments. 
For a subset $Q \subseteq T(N)$
we say that $G$ is \emph{compatible} with $Q$ if every vertex in $Q$
and every vertex in $A_{\neq \emptyset}(I,Q)$ (recall the definition of $A_{\neq \emptyset}(I,Q)$ given in Lemma~\ref{lem:fpt-prepro}) has at least one neighbour in $G$;
note that
if $G$ is compatible with $Q$ then any assignment $\pi$ modelled by $G$
satisfies $\tau^{-1}(a)\neq \emptyset$ for every $a \in A_{\neq
  \emptyset}(I,Q)$.
Intuitively, the graph $G$ captures information about which types of agents are mapped
to which activities (without specifying numbers), while $Q$ captures information about
which agent types are perfectly (i.e., ``completely'') assigned. 

Let $Q \subseteq T(N)$ and let $G$ be a bipartite graph with
bi-partition $\{T(N),A\}$ that is compatible with $Q$. The following simple
lemma shows that, modulo compatibility requirements, 
finding a stable assignment for $I$ can be reduced to finding an individually
rational assignment for $\gamma(I,Q)$  (recall the definition of $\gamma(I,Q)$ given in Lemma~\ref{lem:fpt-prepro}).
\sv{\begin{LEM}[$\star$]}
  \lv{\begin{LEM}}\label{lem:fpt-stable-ind}
  Let $Q \subseteq T(N)$ and let $G$ be a bipartite graph
  with bi-partition $\{T(N),A\}$ that is compatible with $Q$. Then
  for every assignment $\pi : N \rightarrow A^*$ modelled by $G$ and compatible with $Q$,
  it holds that 
  $\pi$ is stable for $I$ if and only if
  $\pi$ is individually rational for $\gamma(I,Q)$. 
\end{LEM}
\begin{proof}
  Since $\pi$ is compatible with $Q$, it follows from
  Lemma~\ref{lem:fpt-prepro} that $\pi$ is stable for $I$ if and only
  if $\pi$ is individually rational for $\gamma(I,Q)$ and $\tau^{-1}(a)\neq
  \emptyset$ for every $a \in A_{\neq \emptyset}(I,Q)$. Consider an arbitrary $a \in
  A_{\neq \emptyset}(I,Q)$. Because $G$ is compatible with $Q$ it
  holds that every $a$ has at least one neighbour in $G$ and since
  $\pi$ is modelled by $G$, we obtain that $\pi^{-1}(a)\neq
  \emptyset$, as required.
\end{proof}

The next, final lemma forms (together with Lemma~\ref{lem:fpt-acyc}) the core component of our proof of Theorem~\ref{the:fpt-agents-act}.

\lv{\begin{LEM}}
\sv{\begin{LEM}[$\star$]}
\label{lem:fpt-acyclic-poly}
  Let $Q \subseteq T(N)$ and let $G$ be an acyclic bipartite graph
  with bi-partition $\{T(N),A\}$ that is compatible with $Q$. Then one
  can decide in time $\bigoh((|N|+|A|)^2|N|^2)$ whether $I$ has a stable assignment
  which is modelled by $G$ and compatible with $Q$.
\end{LEM}
%
\begin{proof}
  By Lemma~\ref{lem:fpt-stable-ind}, $I$ has a stable
  assignment that is modelled by $G$ and compatible with $Q$ if and only if
  $\gamma(I,Q)$ has an individually rational assignment that is modelled by $G$ and compatible
  with $Q$. To determine whether $\gamma(I,Q)$
  has an individually rational assignment that is compatible with $Q$ and modelled by $G$, we
  will employ a reduction to the \GRAPHCON{} problem, which
  can be solved in polynomial-time by Lemma~\ref{lem:graphcon-poly}.
  
  The reduction proceeds as follows. We will construct an 
  instance $I'=(T,\lambda)$ of \GRAPHCON{} that
  is a \yes{}\hy instance if and only if $\gamma(I,Q)=(N,A,(P_n')_{n\in
    N})$ has an individually rational assignment
  that is compatible with $Q$ and modelled by $G$.
  We set $T$ to be the graph $G$ and let $\lambda$ be defined for every $v
  \in V(G)$ as follows:
  \begin{itemize}
  \item if $v \in Q$, then $\lambda(v)=\{|N_v|\}$,
  \item if $v \in T(N) \setminus Q$, then $\lambda(v)=\{0,\dotsc,
    |N_v|-1\}$,
  \item otherwise, i.e., if $v \in A$, then $\lambda(v)=\bigcap_{t \in
      T(N) \land \{v,t\}\in E(G)}P_t'(v)$.
  \end{itemize}
  Note that the reduction can be achieved in time
  $\bigoh(|E(G)||N|)$, assuming that the sets $P_t(a)$ for every $t \in T(N)$
  and $a \in A$ are given in terms of a data structure that allows to
  test containment in constant time; such a data structure could for
  instance be a Boolean array with $|N|$
  entries, whose $i$-th entry is \true{} if and only if $i$ is contained in
  $P_t(a)$. 
  Because the time to construct $\gamma(I,Q)$ and $A_{\neq
    \emptyset}(I,Q)$ is at most $\bigoh((|N|\cdot |A|)|N|)$
  (Lemma~\ref{lem:fpt-prepro}) and the time to solve $I'$ is at most $\bigoh(|V(T)|^2\cdot
    |N|^2)=\bigoh(|V(G)|^2\cdot |N|^2)=\bigoh\big((|N|+|A|)^2|N|^2\big)$
  (Lemma~\ref{lem:graphcon-poly}), we obtain
  $\bigoh\big((|N|+|A|)^2|N|^2\big)$ as the total
  running time of the algorithm.
  
  It remains to show that $I'$ is a \yes\hy
  instance if and only if $I$ has an individually rational assignment that is compatible
  with $Q$ and modelled by $G$.
  Towards showing the forward direction let $\alpha : E(T) \rightarrow
  \Nat$ be a solution for $I'$. We claim that the assignment $\pi : N
  \rightarrow A^*$ that for every edge $e=\{t,a\}$ of $G$ assigns exactly
  $\alpha(\{t,a\})$ agents of type $t$ to activity $a$ and assigns all
  remaining agents (if any) to $a_\emptyset$ is an individually rational assignment for
  $\gamma(I,Q)$ that is compatible with $Q$ and modelled by $G$.
  First observe that for every $t \in T(N)$ and every $a \in A$ it
  holds that $\pi_{t,a}=\alpha(\{t,a\})$ if $\{t,a\} \in E(T)$ and
  $\pi_{t,a}=0$ otherwise. Hence $G_I(\pi)=G=T$ which implies that $\pi$
  is modelled by $G$. We show next that $\pi$ is also compatible
  with $Q$, i.e., $\pi$ satisfies:
  \begin{itemize}
  \item[(P1)] for every $t \in Q$, it holds that $\sum_{a\in
      A}|\pi_{t,a}|=|N_t|$,
  \item[(P2)] for every $t \in T(N) \setminus Q$, it holds that
    $\sum_{a\in A}|\pi_{t,a}|<|N_t|$,
  \end{itemize}
  Towards showing (P1) first note that because $G$ is compatible with
  $Q$, it holds that every $t \in Q$ is adjacent to at least one edge
  in $G$ and thus also in $T$. Moreover, because $\alpha$ is a
  solution for $I'$, we obtain that $\sum_{e=\{t,a\}\in
    E(T)}\alpha(e)=\in \lambda(t)=\{|N_t|\}$. Since $\sum_{e=\{t,a\}\in
    E(T)}\alpha(e)=\sum_{a \in A}|\pi_{t,a}|$, we obtain (P1).

  Towards showing (P2) let $t \in T(N) \setminus Q$. If $t$ is
  isolated in $T$ then $\sum_{a\in A}|\pi_{t,a}|=0$ but $N_t \neq
  \emptyset$ and hence $\sum_{a\in A}|\pi_{t,a}|< |N_t|$, as required.
  If on the other hand $t$ is not isolated in $T$ then because
  $\alpha$ is a solution for $I'$, we obtain that $\sum_{e=\{t,a\}\in
    E(T)}\alpha(e)=\in \lambda(t)=\{0,\dotsc,|N_t|-1\}$. Since $\sum_{e=\{t,a\}\in
    E(T)}\alpha(e)=\sum_{a \in A}|\pi_{t,a}|$, we obtain (P2).

  Finally it remains to show that $\pi$ is individually rational for
  $\gamma(I,Q)$, i.e., for every $a \in A$ and $t \in T(N)$, it holds that
  if $\pi_{t,a}\neq\emptyset$ then $|\pi^{-1}(a)| \in P_t'(a)$. Let $a
  \in A$. If $a$ has no neighbour in $T$
  then $\pi_{t,a}=0$ for every $t \in T(N)$ and the claim holds.
  Hence let $t_1,\dotsc, t_l$ be the neighbours of $a$ in $T$. Then
  because $\alpha$ is a solution for $I'$, we obtain:
 
  \begin{eqnarray*}
    |\pi^{-1}(a)| & = & \sum_{1 \leq i \leq l}|\pi_{t_i,a}|\\
                  & = & \sum_{1 \leq i \leq l}\alpha(\{t_i,a\}) \\
    & \in & \lambda(a)=\bigcap_{1\leq i \leq l}P_{t_i}'(a)
  \end{eqnarray*}
  Hence for every $1 \leq i \leq l$ it holds that
  $|\pi^{-1}(a)| \in P_{t_i}'(a)$, as required.
  
  Towards showing the reverse direction let $\pi$ be an individually
  rational assignment for $\gamma(I,Q)$
  that is compatible with $Q$ and modelled by $G$. We claim that the assignment
  $\alpha: E(G) \rightarrow \Nat$ with $\alpha(\{t,a\})=\pi_{t,a}$ for
  every $\{t,a\}\in E(T)$ is a solution for $I'$.
  Observe that because $\pi$ is modelled by $G$, we have that
  $\pi_{t,a}\neq 0$ if and only if $\{t,a\} \in E(G)$.
  Hence $|\pi_{t}|=\sum_{a \in A}|\pi_{t,a}|=\sum_{\{t,a\}\in
    E(T)}\alpha(t,a)$ for every $t \in T(i)$ and $|\pi^{-1}(a)|=\sum_{t \in
    T(N)}|\pi_{t,a}|=\sum_{\{t,a\}\in E(T)}\alpha(\{t,a\})$.
  Since $\pi$ is compatible with $Q$ we obtain that $\pi_{t}=|N_t|$ for
  every $t \in Q$ and hence $\alpha(t)=|N_t|\in \lambda(t)$. Moreover
  for every $t \in T(N)\setminus Q$ it holds that $\pi_t<|N_t|$ and
  hence $\alpha(t) \in \lambda(t)$.
  Because $\pi$ is an individually rational assignment, it holds that $|\pi^{-1}(a)|\in
  \bigcap_{t \in T(\pi^{-1}(a))}P_t'(a)$ and thus $\sum_{\{t,a\}\in
    E(T)}\alpha(\{t,a\}) \in \lambda(a)$, as required.
\end{proof}

We are now ready to establish
Theorem~\ref{the:fpt-agents-act}.
\begin{proof}[Proof of Theorem~\ref{the:fpt-agents-act}]
  Let $I=(N,A,(P_n)_{n \in N})$ be the given instance of \sGasp{}.
  It follows from Lemma~\ref{lem:fpt-acyc} that it suffices to
  decide whether $I$ has an acyclic stable assignment.
  Observe that every acyclic stable assignment $\pi$ is compatible
  with $\PE(I,\pi)$ and the acyclic bipartite graph $G_I(\pi)$. Hence
  there is an acyclic stable assignment $\pi$ if and only if there is
  a set $Q \subseteq T(N)$ and an acyclic bipartite graph $G$ with
  bi-partition $\{T(N),A\}$ compatible with $Q$ such that
  there is a stable assignment for $I$ modelled by $G$ and compatible with $Q$.

  Consequently, we can determine the existence of a stable assignment
  for $I$ by first branching over every $Q \subseteq T(N)$, then over every acyclic
  bipartite graph $G$ with bi-partition $\{T(N),A\}$ 
  compatible with $Q$, and checking whether $I$
  has a stable assignment that is compatible with $Q$ and modelled by $G$.
  Since there are $2^{|T(N)|}$ many subsets $Q$ of
  $T(N)$ and at most $2^{|T(N)|\cdot |A|}$ (acyclic) bipartite graphs $G$, and we can
  determine whether $I$ has a stable assignment
  compatible with $Q$ and modelled by $G$ in time $\bigoh\big((|N|+|A|)^2|N|^2\big)$ (see
  Lemma~\ref{lem:fpt-acyclic-poly}), it follows that the total running
  time of the algorithm is at most $\bigoh\big(2^{|T(N)|(1+|A|)}(|N|+|A|)^2|N|^2\big)$.
\end{proof}

\section{Result 2: Lower Bound for \sGasp{}}
\label{sec:res2}

In this subsection we complement Theorem~\ref{the:fpt-agents-act} by showing
that if we drop the number of activities in 
the parameterization, then \sGasp{} becomes \W{1}\hy hard. We achieve this via a
parameterized reduction from \SMDPSS{} that we have shown to be
strongly \W{1}\hy hard in Theorem~\ref{the:smdpsshard}.
\begin{THE}
  \sGasp{} is \Weft\emph{[1]}\hy hard parameterized by the number of
  agent types.
\end{THE}
\begin{proof}
  Let $(d,\bar{t},\PPP)$ with $\PPP=(P_1,\dotsc, P_m)$ be an instance
  of \SMDPSS{}. Because \SMDPSS{} is strongly \W{1}\hy hard, we can
  assume that all numbers of the instance $(d,\bar{t},\PPP)$, i.e.,
  the values of all components of the vectors in
  $\{\bar{t}\}\cup\bigcup_{P\in\PPP}P$, are encoded in unary. We will
  also assume that all non-zero components of the vectors
  $\{\bar{t}\}\cup\bigcup_{P\in\PPP}P$ are at least $3$ (this can for
  instance be achieved by multiplying every vector in
  $\{\bar{t}\}\cup\bigcup_{P\in\PPP}P$ with the number $3$).
  We will now construct the instance $I=(N,A,(P_n)_{n \in N})$ 
  of \sGasp{} in polynomial time with $|T(I)|=d+3$ such that
  $(d,\bar{t},\PPP)$ has a solution if and only if  so does $I$. $I$ has one agent type $t_i$ for every $1
  \leq i \leq d$ that comes with $\bar{t}[i]$ agents as well as three
  additional agent types $t_P$, $t_{\neq \emptyset}^1$, and $t_{\neq
    \emptyset}^2$ having one agent each. The last three agent types
  will be employed to ensure that every agent type $t_i$ must be
  perfectly assigned and every activity must be non-empty. Moreover,
  $I$ has one activity $a_\ell$ for every $\ell$ with $1\leq \ell \leq
  m$ as well as one activity $a_P$, which will be used in conjunction
  with the agent type $t_P$ to ensure that all agent types $t_i$ must
  be perfectly assigned.
  The approval set for the agent type $t_i$ w.r.t. an activity
  $a_\ell$ is given by $P_{t_i}(a_\ell)=\SB \bar{v}[i] \SM \bar{v} \in P_\ell
  \land \bar{v}[i]\neq 0 \SE$. Note that because all sets in $\PPP$ are
  simple, it holds that $P_{t_i}(a_\ell)\cap P_{t_j}(a_\ell)=\emptyset$ for
  every $i$, $j$, and $\ell$ with $1\leq i,j\leq d$, $j\neq i$, and
  $1\leq \ell \leq m$ and hence every such activity $a_\ell$ is
  populated by agents of at most one type in any individual rational
  assignment for $I$. Finally, we set $P_{t_i}(a_P)=\{2\}$ for every
  $i$ with $1 \leq i \leq d$, $P_{t_P}(a_P)=\{1,3\}$, $P_{t_{\neq
      \emptyset}^1}(a)=\{1\}$ and $P_{t_{\neq \emptyset}^2}(a)=\{2\}$
  for every activity $a \in A \setminus \{a_P\}$. Note that indeed
  $|T(I)|=d+3$ and $I$ can be constructed in polynomial-time (recall
  that we assumed that all numbers of $(d,\bar{t},\PPP)$ are
  encoded in unary). It remains to show that $(d,\bar{t},\PPP)$ has a
  solution if and only if so does $I$.

  Towards showing the forward direction, let
  $\bar{p}_1,\dotsc,\bar{p}_m$ with $\bar{p}_\ell\in P_\ell$ and
  $\sum_{\ell=1}^m\bar{p}_\ell=\bar{t}$ be a solution for
  $(d,\bar{t},\PPP)$. Let $\pi : N \rightarrow A^*$
  be the assignment defined as follows:
  \begin{itemize}
  \item For every agent type $t_i$ and every vector $\bar{p}_\ell$ such
    that $\bar{p}_\ell[i]\neq \emptyset$, $\pi$ assigns exactly
    $\bar{p}_\ell[i]$ agents to activity $a_j$.
  \item $\pi(n_P)=\{a_P\}$, where $n_P$ is the unique agent of type
    $t_P$,
  \item $\pi(n_{\neq \emptyset}^1)=\pi(n_{\neq
      \emptyset}^2)=a_\emptyset$, where $n_{\neq \emptyset}^1$ and
    $n_{\neq \emptyset}^2$ are the unique agents having type $t_{\neq
      \emptyset}^1$ and $t_{\neq \emptyset}^2$, respectively.
  \end{itemize}
  We claim that $\pi$ is a stable assignment for $I$, which we will
  show using the following sequence of observations:
  \begin{itemize}
  \item[(O1)] Due to the construction of $\pi$, we obtain that $|\pi_{t_i,a_\ell}|=\bar{p}_\ell[i]$
    for every $i$ and $j$ with $1\leq i \leq d$ and $1 \leq \ell\leq m$.
  \item[(O2)] Because of (O1) and the fact that
    $\sum_{\ell=1}^m\bar{p}_\ell=\bar{t}$, we obtain that
    $|\pi_{t_i}|=\bar{t}[i]$ for every $i$ with $1 \leq i \leq
    d$. Since furthermore $I$ has exactly $\bar{t}[i]$ agents of type
    $t_i$, we obtain that the agents of type $t_i$ are perfectly
    assigned by $\pi$.
  \item[(O3)] Because every set in $\PPP$ is simple, it also holds that
    $\pi^{-1}(a_\ell)$ consists of exactly $\bar{p}_\ell[i]$ agents of type
    $i$, where $i$ is the unique non-zero component of $\bar{p}_\ell$.
  \item[(O4)] Since $\pi^{-1}(n_P)=\{n_P\}$ the assignment $\pi$ is stable for the
    agent $n_P$.
  \item[(O5)] Since $\pi(n_{\neq \emptyset}^1)=\pi(n_{\neq
      \emptyset}^2)=a_\emptyset$ and $|\pi^{-1}(a)|\geq 2$ for every
    activity $a \in A\setminus \{a_P\}$ (because of (O3)), it holds that the assignment
    $\pi$ is stable for the agents $n_{\neq \emptyset}^1$ and $n_{\neq
      \emptyset}^2$.
  \item[(O6)] Consider an agent $n$ of type $t_i$. Because of (O2), we
    have that $\pi(n)\neq a_\emptyset$. Hence $\pi(n)=a_\ell$ for some
    $\ell$ with $1 \leq \ell \leq m$ and since
    $|\pi^{-1}(a_\ell)|=\bar{p}_\ell[i]$ (because of (O3)), we have that
    $|\pi^{-1}(a_\ell)| \in P_n(a_\ell)$, which implies that $\pi$ is a
    stable assignment for $n$.
  \end{itemize}
  Consequently, $\pi$ is a stable assignment for $I$.

  Towards showing the reverse direction, let $\pi : N \rightarrow
  A^*$ be a stable assignment for $I$. We start by
  showing that $\pi$ satisfies the following two properties:
  \begin{itemize}
  \item[(P1)] for every $i$ with $1 \leq i \leq d$, all agents of type
    $t_i$ are assigned to some activity in $A \setminus \{a_P\}$,
  \item[(P2)] for every activity $a \in A \setminus \{a_P\}$, it holds
    that $\pi^{-1}(a)\neq \emptyset$ and moreover $T(\pi^{-1}(a))
    \subseteq \{t_1,\dotsc,t_d\}$ and $|T(\pi^{-1}(a))|=1$.
  \end{itemize}
  We will show (P1) and (P2) using the following series of claims that
  hold for any stable assignment $\pi : N \rightarrow A$ for $I$:
  \begin{itemize}
  \item[(C1)] $\pi(n_{\neq \emptyset}^2)=a_\emptyset$ for the unique agent $n_{\neq \emptyset}^2$ of type
    $t_{\neq \emptyset}^2$,
  \item[(C2)] $\pi(n_{\neq \emptyset}^1)=a_\emptyset$ for the unique
    agent $n_{\neq \emptyset}^1$ of type $t_{\neq \emptyset}^1$,
  \item[(C3)] $\pi(n_P)=a_P$ for the unique agent $n_P$ of type
    $t_P$,
  \item[(C4)] $\pi^{-1}(a_P)=\{n_P\}$,
  \end{itemize}

  Towards showing (C1), assume for the contrary that $\pi(n_{\neq \emptyset}^2)=a$ for
  some $a \in A$. Then $a \in A\setminus \{a_P\}$ and
  $|\pi^{-1}(a)|=2$. However, this is not possible since $n_{\neq \emptyset}^2$ is the only
  agent in $I$ that approves size $2$ for any activity in $A \setminus
  \{a_P\}$.

  Towards showing (C2), assume for the contrary that $\pi(n_{\neq \emptyset}^1)=a$ for
  some $a \in A$. Then $a \in A\setminus \{a_P\}$ and
  $|\pi^{-1}(a)|=1$. Moreover, because of (C1), we have that
  $\pi(n_{\neq \emptyset}^2)=a_\emptyset$ and hence $n_{\neq \emptyset}^2$ would prefer $a$ over his
  current assignment, contradicting the stability of $\pi$.

  Towards showing (C3), assume for the contrary that $\pi(n_P) \neq
  a_P$. Then $\pi(n_P)=a_\emptyset$. Moreover, due to the approval set
  of $n_P$, it must hold that $\pi^{-1}(a_P)\neq \emptyset$. Since 
  $P_t(a_P)\in \{\emptyset,\{2\}\}$ for every agent type in
  $T(I)\setminus \{t_P\}$, it follows that $|\pi^{-1}(a_P)|=2$. However,
  this contradicts the stability of $\pi$, since $n_P$ would prefer
  $a_P$ over his current assignment.
  (C4) is a direct consequence of (C3) since $P_{n_P}(a_P)=\{1,3\}$
  and moreover $n_P$ is the only agent with $3 \in P_{n_P}(a_P)$.
  
  We are now ready to show (P1) and (P2). Towards showing (P1) assume
  for a contradiction that there is an agent $n$ whose type is in
  $\{t_1,\dotsc,t_d\}$ such that $\pi(n) \in \{a_\emptyset,a_P\}$.
  Because of (C4), we obtain that $\pi(n)=a_\emptyset$. Moreover, since $2 \in
  P_n(a_P)$ and $|\pi^{-1}(a_P)|=1$ (because of (C4)), it follows that
  $n$ would prefer $a_P$ over his current assignment, which
  contradicts the stability of $\pi$.

  Towards showing (P2) assume for a contradiction that there is an
  activity $a \in A\setminus \{a_P\}$ with
  $\pi^{-1}(a)=\emptyset$. Consider the agent $n_{\neq \emptyset}^1$,
  i.e., the only agent of type $t_{\neq \emptyset}^1$, then because of
  (C2), we have $\pi(n_{\neq \emptyset}^1)=a_\emptyset$. Moreover, since $1 \in
  P_{n_{\neq \emptyset}^1}(a)$ for every $a \in A \setminus \{a_P\}$,
  the agent $n_{\neq \emptyset}^1$
  would prefer $a$ over his current assignment, which contradicts the
  stability of $\pi$. Hence $\pi^{-1}(a)\neq \emptyset$ for every $a
  \in A\setminus \{a_P\}$. Furthermore, since the agent types in
  $\{t_1,\dotsc,t_d,t_{\neq \emptyset}^1,t_{\neq \emptyset}^2\}$ are
  the only types that approve of an activity in $A \setminus \{a_P\}$
  and it follows from (C1) and (C2) that neither $n_{\neq
    \emptyset}^1$ nor $n_{\neq \emptyset}^2$ are assigned to an
  activity in $A \setminus \{a_P\}$, we obtain that
  $T(\pi^{-1}(a))\subseteq \{t_1,\dotsc,t_d\}$. It remains to show
  that $|T(\pi^{-1}(a_\ell))|=1$ for every $1 \leq \ell \leq m$.
  Assume for a contradiction that there are two distinct agent types $t$ and $t'$
  in $\{t_1,\dotsc,t_d\}$ with $t,t' \in T(\pi^{-1}(a_\ell))$. Since $\pi$
  is individual rational, we obtain that $|\pi^{-1}(a_\ell)| \in P_t(a_\ell)$ and
  $|\pi^{-1}(a_\ell)| \in P_{t'}(a_\ell)$. Hence the set $P_\ell$ in $\PPP$
  contains two vectors that share the same value at their non-zero
  component, which contradicts our assumption that $P_\ell$ is a simple
  set. This concludes the proof for (P1) and (P2) and we are now ready
  to complete the proof of the reverse direction.

  Consider an activity $a_\ell$ for some $1 \leq \ell \leq m$. Because of
  (P2), we obtain that all agents in $\pi^{-1}(a_\ell)$ have the same
  type say $t_i$. Because $\pi$ is
  stable it holds that $|\pi^{-1}(a_\ell)|\in P_{t_i}(a_\ell)$ and hence
  there is a vector, say $\bar{p}_\ell$, in $P_\ell$ with
  $\bar{p}_\ell[i]=|\pi^{-1}(a_\ell)|$. We claim that the vectors
  $\bar{p}_1,\dotsc,\bar{p}_m$ chosen in this way form a solution for
  $(d,\bar{t},\PPP)$. Consider a component $i$ with $1 \leq i \leq d$,
  then because of (P1), we obtain that $\sum_{j=1}^m\bar{p}_j[i]$ is
  equal to the number of agents of type $t_i$, which in turn is equal
  to $\bar{t}[i]$ by the construction of $I$. Hence
  $\sum_{j=1}^m\bar{p}_j=\bar{t}$ and $\bar{p}_1,\dotsc,\bar{p}_m$ is
  a solution for $(d,\bar{t},\PPP)$.
\end{proof}

\section{Result 3: \XP\ Algorithms for \sGasp{} and \Gasp{}}
\label{sec:res3}
In this section, we present our \XP{} algorithm for \Gasp{}
parameterized by the number of agent types. In order to obtain this result, 
we observe that the stability of an assignment for
\Gasp{} can be decided by only considering the stability of agents that are assigned to
a ``minimal alternative'' w.r.t.\ their type. We then show that once
one guesses (i.e., branches over) a minimal alternative for every agent type, the problem of
finding a stable assignment for \Gasp{} that is compatible with this
guess can be reduced to the problem of finding a perfect and
individual rational assignment for a certain instance of \sGasp{}, where one additionally
requires that certain activities are assigned to at least one
agent. Our first task will hence be to obtain an \XP\ algorithm which can find such
a perfect and individually rational assignment for \sGasp{}.


\subsection{An \XP\ Algorithm for \sGasp{}}
The aim of this section is twofold. First of all, we obtain Lemma~\ref{lem:xp-poly-IR-a}, which allows us to find certain individually rational assignments in \sGasp{} instances and forms a core part of our \XP\ algorithm for \Gasp{}.

\begin{LEM}
  \label{lem:xp-poly-IR-a}
  Let $I=(N,A,(P_n)_{n \in N})$ be an instance of \sGasp{},
  $Q \subseteq T(N)$, and $A_{\neq \emptyset}\subseteq A$.
  Then one can decide in time $\bigoh(|A|\cdot
  (|N|)^{|T(N)|})$
  whether $I$ has an individual rational assignment $\pi$ that is compatible with
  $Q$ such that $\pi^{-1}(a)\neq \emptyset$ for every $a \in A_{\neq \emptyset}$.
\end{LEM}
\begin{proof}
  Let $k=|T(N)|$ and $r=\max_{t\in T(N)}|N_t|$.  
  We construct an instance $I'=(k, \PPP, \bar{t})$ of \MDPSS{} such
  that $I'$ is a \yes{}\hy instance if and only if
  $I=(N,A,(P_n)_{n\in N})$ has an
  individually rational assignment $\pi : N \rightarrow A^*$ that is
  compatible with $Q$ and $\pi^{-1}(a)\neq \emptyset$ for every $a
  \in A_{\neq \emptyset}$.
  Let
  $T(N)=\{t_1,\dotsc,t_k\}$. The set $\PPP$ contains
  one set $P_a$ for every activity $a \in A$, defined as
  follows. For every number $p \in \bigcup_{t \in T(N)}P_t(a)$ the set
  $P_a$ contains the set of all vectors $\bar{p} \in [r]_0^k$ such that
  $|\bar{p}|=p$ and $\bar{p}[i]=0$ for every $i$ with $1\leq i \leq k$
  such that $p \notin P_{t_i}(a)$. Moreover, if $a \notin A_{\neq
    \emptyset}$, then the set $P_a$ additionally contains the
  all-zero vector $\bar{0}$. This completes the construction of $I'$.

  We claim that $I$ has an
  individually rational assignment $\pi : N \rightarrow A^*$ that is
  compatible with $Q$ and $\pi^{-1}(a)\neq \emptyset$ for every $a
  \in A_{\neq \emptyset}$ if and only if the solution $T$ for
  $I'$ contains a vector $\bar{t}$ with $\bar{t}[i]=|N_{t_i}|$ for every
  $t_i \in Q$ and $\bar{t}[i]<|N_{t_i}|$ otherwise (i.e. for every
  $t_i \in T(N)\setminus Q$). Note that establishing this claim completes the proof of
  the lemma since $I'$ can be constructed in time $\bigoh(r^k)$ and
  solved in time $\bigoh(|A|\cdot r^k)$ by Lemma~\ref{lem:solve-mdpss}.

  Towards showing the forward direction, let $\pi : N \rightarrow A^*$
  be an individually rational assignment for $I$ that is
  compatible with $Q$ and $\pi^{-1}(a)\neq \emptyset$ for every $a
  \in A_{\neq \emptyset}$. For every $a \in A$ let $\bar{p}_a$ be
  the vector with $\bar{p}_a[i]=|\pi_{t_i,a}|$ for every $i$ with $1
  \leq i \leq k$. Note that if $\bar{p}_a\neq \bar{0}$ then 
  $\bar{p}_a \in P_a$ for every $a \in A$ because $\pi$ is individually
  rational. On the other hand, if $\bar{p}_a=\bar{0}$ then $a \notin A_{\neq
    \emptyset}$ and so we also obtain $\bar{p}_a \in P_a$. 
  Hence the vector $\bar{t}=\sum_{a\in A}\bar{p}_a$ is in the solution
  for $I'$ and moreover $\bar{t}[i]=\sum_{a\in A}\bar{p}_a[i]=\sum_{a\in
    A}|\pi_{t_i,a}|=|\pi_{t_i}|$ for every $i$ with $1 \leq i \leq k$.
  Finally, because $\pi$ is compatible with
  $Q$, we obtain that $\bar{t}[i]=|\pi_{t_i}|=|N_{t_i}|$ for every $i$ with $t_i
  \in Q$ and also $\bar{t}[i]=|\pi_{t_i}|<|N_{t_i}|$ for every $i$ with $t_i
  \in T(N) \setminus Q$, as required.  

  Towards showing the reverse direction, assume that the solution $T$
  for $I'$ contains a vector $\bar{t}$ with $\bar{t}[i]=|N_{t_i}|$ for every
  $t_i \in Q$ and $\bar{t}[i]<|N_{t_i}|$ otherwise (i.e., for every
  $t_i \in T(N)\setminus Q$) and for every $a \in A$ let $\bar{p}_a$ be the vector in
  $P_a$ such that $\sum_{a \in A}\bar{p}_a=\bar{t}$. We claim that the
  assignment $\pi : N \rightarrow A^*$ that for every $1 \leq i \leq
  k$ and every $a \in A$ assigns exactly $\bar{p}_a[i]$ agents of type
  $t_i$ to activity $a$ and all remaining agents to activity
  $a_\emptyset$ is an individually rational assignment for $I$ that
  is compatible with $Q$ and $\pi^{-1}(a)\neq \emptyset$ for every $a
  \in A_{\neq \emptyset}$. First observe that for every $i$ with
  $1 \leq i \leq k$, it holds that
  $|\pi_{t_i}|=\sum_{a\in A}|\pi_{t_i,a}|=\sum_{a\in
    A}\bar{p}_a[i]=\bar{t}[i]$. Note also that because $\bar{t}[i]\leq
  |N_{t_i}|$, i.e., there is a sufficient number of agents for every
  agent type, we know that it is possible to assign the agents
  according to $\pi$. Since in addition it holds that
  $\bar{t}[i]=|N_{t_i}|$ if $t_i \in Q$ and $\bar{t}[i]<|N_{t_i}|$
  if $t_i\not \in Q$, it follows that
  $\pi$ is compatible with $Q$. Moreover, because $\bar{0}\notin P_a$ for every
  $a \in A_{\neq \emptyset}$, it also holds that
  $|\pi^{-1}(a)|=\sum_{1\leq i \leq k}|\pi_{t_i,a}|=\sum_{1\leq i \leq
    k}\bar{p}_a[i]\neq 0$ for every such $a$. It remains to show that $\pi$
  is individually rational for $I$. By the definition of $\pi$ it
  holds that whenever $\pi$ assigns an agent of type $t_i$ to
  some activity $a\in A$, then $\bar{p}_a[i]\neq 0$. Moreover, by the
  construction of $I'$ it holds that if $\bar{p}_a[i]\neq 0$ then $|\bar{p}_a|\in
  P_{t_i}(a)$. Hence because $|\pi^{-1}(a)|=|p_a|$, we obtain that
  $|\pi^{-1}(a)| \in P_{t_i}(a)$.
\end{proof}

As a secondary result, we can already obtain an \XP\ algorithm for \sGasp{} parameterized by the number of agent types. This may also be of interest, as the obtained running time is strictly better than that of the algorithm obtained for the more general \Gasp{}.
The last thing we need for this result is the following corollary, obtained as a direct consequence of Lemma~\ref{lem:fpt-prepro} and Lemma~\ref{lem:xp-poly-IR-a}.

\begin{COR}
  \label{cor:xp-poly-IR}
  Let $I=(N,A,(P_n)_{n \in N})$ be an instance of \sGasp{} and
  $Q \subseteq T(N)$. Then one can decide in time $\bigoh(|N|^2\cdot |A|+|A|\cdot
  (|N|)^{|T(N)|})$
  whether $I$ has a stable assignment compatible with $Q$.
\end{COR}

We can now prove the following.

\begin{THE}\label{the:xp}
  An instance $I=(N,A,(P_n)_{n\in N})$ of \sGasp{} can be solved in time
  $|A|\cdot |N|^{\bigoh(|T(N)|)}$.
\end{THE}

\begin{proof}
  Let $I=(N,A,(P_n)_{n \in N})$ be the given instance of \sGasp{}.
  The algorithm loops through every $Q \subseteq T(N)$ and 
  in each branch checks whether $I$ has a
  stable assignment that is compatible with $Q$. Note that due to
  Corollary~\ref{cor:xp-poly-IR} this can be achieved in time 
  $\bigoh(|N|^2\cdot |A|+|A|\cdot
  (|N|)^{|T(N)|})$ for every $Q
  \subseteq T(N)$. Since there are $2^{|T(N)|}$ subsets of
  $T(N)$, we obtain $\bigoh(2^{|T(N)|}(|N|^2\cdot |A|+|A|\cdot
  (|N|)^{|T(N)|}))=|A|\cdot |N|^{\bigoh(|T(N)|)}$ as the total running time of the algorithm.
\end{proof}

\subsection{An \XP{} algorithm for \Gasp{}}
\label{sec:algm}

Our aim here is to use Lemma~\ref{lem:xp-poly-IR-a} to
obtain an \XP\ algorithm for \Gasp{}. To simplify the presentation of our 
algorithm, we start by introducing
the notion of an NS$^*$-deviations that combines and unifies
individual rationality and NS-deviations. Namely, let
$I=(N,A,(\pref_n)_{n\in N})$ be a \Gasp{} instance, $\pi : N
\rightarrow A^*$ be an assignment, and $n \in N$. Then we say that $n$
has an NS$^*$-deviation to an activity $a' \in A^* \setminus
\{\pi(n)\}$ if $(a',|\pi^{-1}(a')|+1)\prefs_{T(n)} (a,|\pi^{-1}(a)|)$. Note
that in order to deal with the case that $a'=a_\emptyset$, we define
$(a_\emptyset,i+1)$ to be equal to $(a_\emptyset,1)$ for every $i$.
\begin{OBS}
  \label{obs:deviation}
  An assignment $\pi$ for $I$ is
  stable if and only if no agent $n\in N$ has an NS$^*$-deviation to
  any activity in $A^*\setminus \{\pi(n)\}$.
\end{OBS}
Let $I$ and $\pi$ be as above and let $t \in T(I)$. We denote by
$\pi_t^*$ the set of activities $\pi_t$ if $t$ is
perfectly assigned by $\pi$ and $\pi_t \cup \{a_\emptyset\}$, otherwise.
We say an activity $a\in \pi_t^*$ is \emph{minimal with respect to
  $t$} if $(a',|\pi^{-1}(a')|)\pref_t (a,|\pi^{-1}(a)|)$ for each
$a'\in \pi_t^*$ and we address the alternative $(a,|\pi^{-1}(a)|)$ as
a \emph{minimal alternative} with respect to $t$. The following lemma
provides our first key insight, by showing that only agents assigned
to a minimal activity need to be checked for having an NS-deviation.

\begin{LEM}
  \label{lem:min_alternative}
  Let $\pi$ be an assignment and let $n$ be an agent of type $t$ that
  has an NS$^*$-deviation. Then there is an agent $n'$ of type $t$
  that also has an NS$^*$-deviation and is assigned to a
  minimal activity w.r.t. $t$.
\end{LEM}
\begin{proof}
  Let $a=\pi(n)$, then because $n$ has an NS$^*$-deviation, there is an
  activity $a' \in A^*\setminus \{a\}$ such that
  $(a',|\pi^{-1}(a')|+1)\prefs_{t} (a,|\pi^{-1}(a)|)$. Let $a_m$ be a
  minimal activity w.r.t. $t$. Then because
  $(a,|\pi^{-1}(a)|)\pref_t(a_m,|\pi^{-1}(a_m|)$ any agent assigned to
  $a_m$ does also have an NS$^*$-deviation to $a_m$.
\end{proof}
The following lemma is a direct consequence of
Observation~\ref{obs:deviation} and Lemma~\ref{lem:min_alternative}
and allows us to characterize the stability condition of an assignment
in terms of minimial activities for each agent type.
\begin{LEM}
  \label{lem:min-property}
  An assignment $\pi$ for $I$ is stable if and only if for each $t\in
  T(N)$ and each $a \in A^*\setminus \{a_m\}$, it holds that $(a_m,|\pi^{-1}(a_m)|) \pref_t
  (a,|\pi^{-1}(a)|+1)$, where $a_m$ is a minimal activity w.r.t.\ $t$.
\end{LEM}
\begin{proof}
  It follows from Lemma~\ref{lem:min_alternative} and
  Observation~\ref{obs:deviation} that $\pi$ is not stable
  if and only if there is an agent $n$ participating in a minimal activity  $a_m$ w.r.t.\ $T(n)$ that has an NS$^*$-deviation. This
  implies that $(a,|\pi^{-1}(a)|+1)\prefs_{T(n)}
  (a_m,|\pi^{-1}(a_m)|)$ for some $a \in A^*\setminus \{a_m\}$.
\end{proof}

The following theorem now employs the above lemma to construct an
instance $I'$ of \sGasp{} together with a subset $A_{\neq \emptyset}$
of activities such that for every function $f_{\min} : T(I)
\rightarrow X$ (or in other words for every guess of minimal alternatives
in an assignment), it holds that $I$ has a stable assignment such that
$f_{\min}(t)$ is a minimal alternative w.r.t.\ $t$ for every $t \in T(I)$
if and 
only if $I'$ has a perfect and individual rational assignment $\pi$
such that $\pi^{-1}(a)\neq \emptyset$ for every $a \in A_{\neq \emptyset}$. For
brevity, we will say that an assignment $\pi$ is \emph{compatible with $f_{\min}$}
if and only if $f_{\min}(t)$ is a minimal alternative w.r.t.\ $t$ for every
$t \in T(I)$.
\begin{THE}
  \label{thm:GtosG}
  Let $I=(N,A,(\pref_n)_{n\in N})$ be an instance of \Gasp{} and let
  $f_{\min}:T(N)\rightarrow X$, which informally represents a guess of a minimal
  alternative for every agent type. Then one can in time $\bigoh(|N|^2|A|)$ construct an instance
  $I'=(N,A\cup\{a_{\phi}\},(P_n)_{n\in N})$ of \sGasp{} together with
  a subset $A_{\neq \emptyset}$ of activities such that
  $|T(I')|\leq 2|T(I)|$ and $I$ has a stable assignment
  compatible with $f_{\min}(t)$ if and
  only if $I'$ has a perfect individual rational assignment $\pi$ with
  $\pi^{-1}(a)\neq \emptyset$ for every $a \in A_{\neq \emptyset}$.
\end{THE}
\begin{proof}
  As a first step we remove all alternatives from $X$ that violate the
  stability condition of an assignment $\pi$ for $I$ having $f_{\min}(t)$ as a
  minimal alternative w.r.t. $t$ for every $t \in T(I)$ given in
  Lemma~\ref{lem:min-property}. Namely, let $X'$
  be the set of alternatives obtained from $X\setminus \{(a_\emptyset,1)\}$ after removing all
  alternatives $(a,i) \in X$ for which there is an agent type $t \in
  T(I)$ such that $(a,i+1) \prefs_t f_{\min}(t)$. 
  We will also assume
  w.l.o.g. that $f_{\min}(t) \pref_t (a_\emptyset,1)$ and if
  $(a_\emptyset,1)\pref_t f_{\min}(t) \pref_t (a_\emptyset,1)$, then
  $f_{\min}(t)=(a_\emptyset,1)$.
  
  Note that the activity $a_\phi$ of $I'$ will take the place of the activity
  $a_\emptyset$ of $I$. We define the preferences $(P_n)_{n \in N}$
  for $I'$ as follows. Let $t \in T(I)$, let $(a_m,i_m)=f_{\min}(t)$
  and let $n_t \in N_t$ be an arbitrary agent of type
  $t$. Then we set $P_{n_t}=\{(a_m,i_m)\}$ if $a_m\neq a_\emptyset$ and
  $P_{n_t}=\{(a_\emptyset,j) \SM 1 \leq j \leq |N| \SE$
  otherwise. Informally, this will ensure that
  $\pi(n_t)=a_m$ and $(a_m,\pi^{-1}(a_m))=(a_m,i_m)$ if $a_m\neq
  a_\emptyset$ and $\pi(n_t)=a_\phi$ otherwise for any perfect 
  assignment $\pi$ for $I'$.

  Moreover, for every other agent
  $n$ of type $t$, i.e., $n \in N_t$ and $n\neq n_t$, we set
  $P_{n}=\SB (a,i)\in X' \SM (a,i) \pref_t f_{\min}(t)\SE$ if $a_m\neq
  a_\emptyset$ and 
  $P_{n}=\SB (a,i)\in X' \SM (a,i) \pref_t f_{\min}(t)\SE \cup \SB (a_\phi,j) \SM
  1 \leq j \leq |N| \SE$ otherwise. Informally, this ensures that
  only alternatives that are at least as preferred as $f_{\min}(t)$
  and that if $a_m=a_\emptyset$, then the agents can be assigned to
  $a_\phi$ for any size. Finally, the set $A_{\neq \emptyset}$ contains all activities $a
  \in A$ for which there is an agent type $t$ such that $(a,1)\prefs_t
  f_{\min}(t)$, i.e., all activities that must be non-empty for any
  stable assignment $\pi$ for $I$.
  
  Clearly the construction of $I'$ can be achieved in time $\bigoh(|N|^2|A|)$
  and moreover $|T(I')|\leq 2|T(I)|$. It remains to show that for
  every choice of $f_{\min}$, $I$ has a stable assignment $\pi$
  compatible with $f_{\min}(t)$ if
  and only if $I'$ has a perfect individual rational assignment $\pi$
  with $\pi^{-1}(a)\neq \emptyset$ for every $a \in A_{\neq \emptyset}$.

  Towards showing the forward direction, let $\pi$ be a stable
  assignment for $I$ compatible with $f_{\min}(t)$. Note that,
  w.l.o.g., we can assume that $\pi(n_t)=a_m$, where
  $(a_m,i_m)=f_{\min}(t)$ for every $t \in T(I)$ since otherwise we
  can switch the assignment for $n_t$ with some agent that is assigned
  to $a_m$, which exists because $\pi$ is compatible with $f_{\min}$.
  We claim that $\pi_1 : N \rightarrow (A \cup
  \{a_\phi\})^*$ such that for every $n\in N$, $\pi_1(n)=\pi(n)$ if $\pi(n)\neq
  a_\emptyset$ and $\pi_1(n)=a_\phi$ otherwise, is a perfect
  individual rational assignment for $I'$ such that $\pi_1^{-1}(a)\neq
  \emptyset$ for every $a \in A_{\neq \emptyset}$. The assignment $\pi_1$ is obviously
  perfect. 
  Towards showing that $\pi_1$ is individual rational assume for a
  contradiction that there is an agent $n\in N$ say of type $t$ with $\pi_1(n)=a\neq
  a_\emptyset$ such that $|\pi_1^{-1}(a)| \notin P_n(a)$. If $n=n_t$,
  then by our assumption on $\pi$, we have that $\pi(n)=a_m$ and
  $|\pi^{-1}(a_m)|=i_m$, where $(a_m,i_m)=f_{\min}(t)$. Hence, if
  $a_m\neq a_\emptyset$ then $\pi_1(n)=a_m$ and
  $|\pi_1^{-1}(a_m)|=i_m$, which because $P_{n_t}=\{(a_m,i_m)\}$, implies that $|\pi_1^{-1}(a_m)|=i_m \in
  P_n(a_m)$. If, on the other hand, $a_m=a_\emptyset$, then
  $\pi_1(n)=a_{\phi}$, which because $P_{n_t}=\SB (a_\phi,j) \SM 1
  \leq j \leq |N|\SE$ implies that $|\pi_1^{-1}(a_m)|=|\pi_1^{-1}(a_\phi)|
  \in P_n(a_\phi)$. Finally consider the case that $n\neq n_t$. Note that if $\pi_1(n)=a_\phi$, then
  $\pi(n)=a_\emptyset$, which
  implies that $(a_\emptyset,1)=f_{\min}(t)$ and hence because $P_n=\SB
  (a_\phi,j) \SM 1 \leq j \leq |N|\SE$, we obtain
  $(a_\phi,|\pi_1^{-1}(a_\phi)|) \in P_n$.
  Moreover, if $\pi_1(n)=a\neq a_\phi$, then
  $|\pi_1^{-1}(a)| \notin P_n(a)$ implies that either
  $(a,|\pi_i^1(a)|) \notin X'$ or $f_{\min}(t) \pref_t
  (a,|\pi_1^{-1}(a)|)=(a,|\pi^{-1}(a)|)$. In the former case $\pi$ is
  not stable and in the latter case $\pi$ is not compatible with $f_{\min}$.
  Finally, it remains to show that $\pi_1^{-1}(a)\neq \emptyset$ for
  every $a \in A_{\neq \emptyset}$. Assume for a contradiction that there is an
  activity $a \in A_{\neq \emptyset}$ such that
  $\pi_1^{-1}(a)=\emptyset$. Note that then also
  $\pi^{-1}(a)=\emptyset$ and because of the definition of $A_{\neq
    \emptyset}$ it follows that there is an agent type $t$ such that
  $(a,1)\prefs_t f_{\min}(t)=(a_m,i_m)$. But this contradicts the stability of
  $\pi$ since now the at least one agent assigned to the minimal
  activity $a_m$ has an NS$^*$-deviation to $a$ in $I$.

  Towards showing the reverse direction, let $\pi_1$ be a perfect
  individual rational assignment for $I'$ such that $\pi_1^{-1}(a)\neq
  \emptyset$ for every $a \in A_{\neq \emptyset}$. We claim that the assignment $\pi$ with
  $\pi(n)=\pi_1(n)$ if $\pi_1(n)\neq a_\phi$ and $\pi(n)=a_\emptyset$
  otherwise is a stable assignment for $I$ that is compatible with
  $f_{\min}$; note that because $\pi_1$ is perfect it cannot be the
  case that $\pi_1(n)=a_\emptyset$. We start by showing that $\pi$ is compatible with
  $f_{\min}$, for every $t \in T(I)$. We distinguish two cases
  depending on whether or not $(a_\emptyset,1)=f_{\min}(t)$. If
  $(a_\emptyset,1)=f_{\min}(t)$, then $P_{n_t}=\SB (a_\phi,j) \SM 1
  \leq j \leq |N|\SE$ and hence because $\pi_1$ is perfect, we obtain
  that $\pi_1(n_t)=a_\phi$. It follows that $\pi(n_t)=a_\emptyset$,
  which implies that $f_{\min}=(a_\emptyset,1)$ is activated in
  $\pi$. Finally, because $P_n=\SB (a,i)\in X'
  \SM (a,i) \pref_t f_{\min}(t)\SE \cup \SB (a_\phi,j) \SM 1 \leq j
  \leq |N| \SE$ for every $n \in N_t$ with $n\neq n_t$, it holds that
  if $\pi_1(n)=a\neq a_\phi$, then $(a,|\pi^{-1}_1(a)|)=(a,|\pi^{-1}(a)|) \pref_t
  f_{\min}(t)$, as required. If, on the other hand,
  $(a_m,i_m)=f_{\min}(t)$ and $a_m\neq a_\emptyset$, then because
  $P_{n_t}=\{(a_m,i_m)\}$, we have that $\pi(n_t)=a_m$ and
  $|\pi^{-1}(a_m)|=i_m$ and hence $f_{\min}(t)$ is activated in
  $\pi$. Moreover, for every $n \in N_t$ with $n\neq n_t$, we obtain
  that $\pi_1(n)\neq a_\emptyset$ and $\pi_1(n) \neq a_\phi$,
  because $\pi_1$ is perfect and $P_n=\SB (a,i)\in X'\SM
  (a,i) \pref_t f_{\min}(t)\SE$. Hence
  $(\pi(n),|\pi^{-1}(\pi(n))|)=(\pi_1(n),|\pi_1^{-1}(\pi_1(n))|)$ and since
  $(\pi_1(n),|\pi^{-1}_1(\pi_1(n))|) \in \SB (a,i)\in X'\SM (a,i) \pref_t
  f_{\min}(t)\SE$, we obtain that $(\pi(n),|\pi^{-1}(\pi(n))|) \pref_t
  f_{\min}(t)$, as required.
  
  Towards showing that $\pi$ is also stable assume for a contradiction
  that this is not the case, then because of
  Lemma~\ref{lem:min-property} there is an agent type $t$ and an
  activity $a \in A^*\setminus \{a_m\}$ such that
  $(a,|\pi^{-1}(a)|+1)\prefs_t f_{\min}=(a_m,i_m)$. Note that if
  $|\pi^{-1}(a)|=0$, then $a \in A_{\neq \emptyset}$ contradicting
  our assumption that $\pi^{-1}_1(a)\neq \emptyset$ for every $a \in
  A_{\neq \emptyset}$. Moreover, if $a=a_\emptyset$, then this
  contradicts our assumption that $f_{\min}(t) \pref_t (a_\emptyset,1)$.
  Hence $a \neq a_\emptyset$ and $|\pi^{-1}(a)|\neq 0$, which implies
  that there is an agent $n$ with $\pi(n)=a$ and furthermore
  $\pi(n)=\pi_1(n)$. But then, because $\pi_1$ is individual rational,
  we obtain that $(|\pi_1^{-1}(a)|) \in P_n(a)$, which implies that
  $(a,|\pi_1^{-1}(a)|) \in X'$ contradicting our choice of $X'$.
\end{proof}

We can now proceed to the main result of this section.

\begin{THE}
  \label{thm:GaspXP}
  An instance $I=(N,A,(\pref_{n})_{n\in N})$ of \Gasp{} can be solved
  in time $(|A|\cdot |N|)^{\bigoh(|T(I)|)}$. 
\end{THE}
\begin{proof}
  Given an instance $I=(N,A,(\pref_{n})_{n\in N})$ of \Gasp{}, the
  algorithm enumerates all of the at most $(|A|\cdot|N|)^{|T(I)|}$
  possible functions $f_{\min}$ and for each such function $f_{\min}$
  the algorithm uses Theorem~\ref{thm:GtosG} to construct the instance
  $I'=(N,A\cup\{a_\phi\},(P_n)_{n \in N})$ of \sGasp{} with
  $|T(I')|\leq |T(I)|$ together with
  the set $A_{\neq \emptyset}$ of activities in time
  $\bigoh(|N|^2|A|)$.
  It then uses Lemma~\ref{lem:xp-poly-IR-a} to decide
  whether $I'$ has a perfect individual rational assignment $\pi_1$
  such that $\pi_1^{-1}(a)\neq \emptyset$ for every $a \in A_{\neq
    \emptyset}$ in time $\bigoh((|A|+1)(|N|)^{|T(I')|})=\bigoh((|A|+1)(|N|^{2|T(I)|})$. If this is true for at least one of the functions
  $f_{\min}$, the algorithm returns that $I$ has a solution, otherwise
  the algorithm correctly returns that $I$ has no solution. The total
  running time of the algorithm is hence $\bigoh(|A|^{|T(I)|}\cdot |N|^{2|T(I)|})=(|A|\cdot|N|)^{\bigoh(|T(I)|)}$.
\end{proof}

\section{Result 4: Lower Bound for \Gasp{}}
\label{sec:res4}

This section presents our hardness result for \Gasp{}.
In particular, we show that \Gasp{} is unlikely to be
fixed-parameter tractable parameterized by both the number of
activities and the number of agent types.

\sv{\begin{THE}[$\star$]}
  \lv{\begin{THE}}\label{the:hard-gasp}
  \Gasp{} is \Weft\emph{[1]}\hy hard parameterized by the number of activities
  and the number of agent types.
\end{THE}
\begin{proof}
  We will employ a
  parameterized reduction from the \mc{} problem, which is
  well-known to be \W{1}\hy complete~\cite{Pietrzak03}.

  \pbDefP{\mc{}}
  {An integer $k$, a $k$-partite graph $G=(V,E)$ with partition $\{V_1,\ldots,V_k\}$ of $V$ into sets of equal size.}
  {$k$}
  {Does $G$ have a $k$-clique, i.e., a set \mbox{$C\subseteq V$} of $k$ vertices such that $\forall u,v\in C$, with $u\neq v$ there is an edge $\{u,v\}\in E$?}
  
  We denote by $E_{i,j}$ the set of edges of $G$ that have one endpoint
  in $V_i$ and one endpoint in $V_j$ and we assume w.l.o.g. that
  $|V_i|=n$ and $|E_{i,j}|=m$ for every $i$ and $j$ with $1 \leq i < j
  \leq k$ (see the standard parameterized complexity textbook
  for a justification of these assumptions~\cite{CyganFKLMPPS15}). 

  \sloppypar Given an instance $(G,k)$ of \mc{} with partition $V_1,\dotsc,V_k$,
  we construct an equivalent instance $I=(N,A,(\pref_n)_{n\in N})$ of \Gasp{} in
  polynomial time with $\binom{k}{2}+k$ activities and $2k+1$ agent types.

  The instance $I$ has the following activities:
  \begin{itemize}
  \item For every $i$ with $1 \leq i \leq k$ the activity $a_i$, whose
    size in a stable assignment for $I$ will be used to identify the
    vertex in $V_i$ chosen to be part of a $k$-clique in $G$. 
  \item For every $i$ and $j$ with $1 \leq i < j \leq k$ the activity
    $a_{i,j}$, whose size in a stable assignment for $I$ will be used
    to identify the edge in $E_{i,j}$ chosen to be part of a
    $k$-clique in $G$. 
  \end{itemize}
  For every $i$ and $j$ with $1 \leq i < j \leq k$ let $\alpha_i$ be a
  bijection from $V_i$ to the set $\{3,5,\dotsc, 2(n-1)+1,2n+1\}$
  and similarly let $\alpha_{i,j}$ be a bijection
  from $E_{i,j}$ to the set $\{1,3,\dotsc, 2m-1\}$.
  The
  main ideas behind the reduction are as follows. First the reduction
  ensures that for every stable assignment $\pi : N
  \rightarrow A^*$ for $I$ the size $s=|\pi^{-1}(a_i)|$ of activity $a_i$
  uniquely identifies a vertex in $V_i$, i.e.,
  the vertex $\alpha_i^{-1}(s)$, and the size $s=|\pi^{-1}(a_{i,j})|$ of activity $a_{i,j}$
  uniquely identifies an edge in $E_{i,j}$, i.e.,
  the edge $\alpha_{i,j}^{-1}(s)$. Employing a set of ``special
  agents'' and their associated preference lists, the reduction will
  then ensure that the vertices identified by the sizes of the
  activities $a_1,\dotsc,a_k$ are endpoints of all the edges
  identified by the sizes of the activities
  $a_{1,2},a_{1,3},\dotsc,a_{k-1,k}$, which implies that these
  vertices form a $k$-clique in $G$.

  For a pair $i,j$ of numbers, we denote by $o(i,j)$, the (ordered)
  pair $i,j$ if $i\leq j$ and the (ordered) pair $j,i$ otherwise.
  We will now introduce the equivalence classes required for the
  definition of our preference lists. Namely, we define the following
  equivalence classes:
  \begin{itemize}
  \item For $x \in \{1,2\}$ we define the set $C_A^x=\SB (a_i,x) \SM 1 \leq i
    \leq k \SE$,
  \item For every $i$ with $1\leq i \leq k$, we define the following
    sets:
\begin{itemize}
    \item $C_V(i)=\SB (a_i,\alpha_i(v)) \SM v \in V_i\SE$,
    \item $C_V^{+1}(i)=\SB (a_i,\alpha_i(v)+1) \SM v \in V_i\SE$,
    \end{itemize}
  \item For every $i$ and $j$ with $1\leq i < j \leq k$, we define the
    set $C_E(i,j)=\SB (a_{i,j},\alpha_{i,j}(e)) \SM e \in E_{i,j}\SE$,
  \item For every $i$ and $v \in V_i$, we define the set $C_I(i,v)=\SB
    (a_i,\alpha_i(v))\SE \cup \SB (a_{o(i,j)}, \alpha_{o(i,j)}(e)+1) \SM 1
    \leq j \leq k \land j\neq i \land e \in E_{i,j} \land v \in e \SE$, i.e.,
    $C_V(i,v)$ contains the tuple $(a_i,\alpha(v))$ and all tuples
    $(a_{o(i,j)},\alpha_{o(i,j)}(e)+1)$ such that $j\neq i$ and the
    edge $e \in E_{o(i,j)}$ is incident to $v$.
  \end{itemize}

  Now ready to define the required preference lists.
  When defining a preference list we will only list the equivalence
  classes that are more or equally preferred to the alternative
  $(a_\emptyset,1)$ and assume
  that all remaining alternatives, i.e., all alternatives that are not
  listed, are less preferred than $(a_\emptyset,1)$.
  \begin{itemize}
  \item The \emph{validity} preference list, denoted by $P_{\textup{VAL}}$,
    defined as $C_V \cup C_E > C_A^2 > C_A^1 >
    (a_\emptyset,1)$. Informally, $P_{\textup{VAL}}$ is crucial in
    ensuring that $|\pi^{-1}(a_i)| \in \SB \alpha_i(v) \SM v \in
    V_i\SE$ and $|\pi^{-1}(a_{i,j})| \in \SB \alpha_{i,j}(e) \SM e \in
    E_{i,j}\SE$ for every stable assignment $\pi$ for $I$ and every
    $i$ and $j$ with $1 \leq i<j \leq k$.
  \item For every $i$ with $1 \leq i \leq k$, let $v_1,\dotsc,v_u$ be the
    unique ordering of the vertices in $V_i$ in ascending order
    w.r.t. $\alpha_i$. We define the following two preference lists
    for every $i$ with $1 \leq i \leq k$:
    \begin{itemize}
    \item The \emph{forward-vertex} preference list, denoted by $P_V^{\rightarrow}(i)$,
      defined as $C_V^{+1}(i) > C_I(i,v_u) > C_I(i,v_{n-1}) > \dotsb > C_I(i,v_1) >
      (a_\emptyset,1)$. Informally, $P_V^{\rightarrow}(i)$ is crucial
      to ensure that for every $j$ with $1 \leq j \leq k$ and $j \neq
      i$ the edge $e$ with $\alpha_{o(i,j)}(e)=|\pi^{-1}(a_{o(i,j)})|$
      is not adjacent with any vertex $v' \in V_i$ such that
      $\alpha_i(v')>\alpha_i(v)$ for the vertex $v$ with
      $\alpha_i(v)=|\pi^{-1}(a_i)|$. This intuition will be made
      precise in Claim~\ref{clm:hard-gasp-pref}.
      
    \item The \emph{backward-vertex} preference list, denoted by $P_V^{\leftarrow}(i)$,
      is defined as $C_V^{+1}(i) > C_I(i,v_1) > C_I(i,v_2) > \dotsb > C_I(i,v_u) >
      (a_\emptyset,1)$.
      Informally, $P_V^{\leftarrow}(i)$ is crucial
      to ensure that for every $j$ with $1 \leq j \leq k$ and $j \neq
      i$ the edge $e$ with $\alpha_{o(i,j)}(e)=|\pi^{-1}(a_{o(i,j)})|$
      is not adjacent with any vertex $v' \in V_i$ such that
      $\alpha_i(v')<\alpha_i(v)$ for the vertex $v$ with
      $\alpha_i(v)=|\pi^{-1}(a_i)|$. This intuition will be made
      precise in Claim~\ref{clm:hard-gasp-pref}.
    \end{itemize}
    Informally, $P_V^{\rightarrow}(i)$ and $P_V^{\leftarrow}(i)$
    together ensure that for every $j$ with $1 \leq j \leq k$ and $j \neq i$, 
    the edge $e$ with $\alpha_{o(i,j)}(e)=|\pi^{-1}(a_{o(i,j)})|$
    is adjacent with the vertex $v$ with
    $\alpha_i(v)=|\pi^{-1}(a_i)|$. This intuition will be made
    precise in Claim~\ref{clm:hard-gasp-pref}
  \end{itemize}
  
  We are now ready to define the set $N$ of agents:
  \begin{itemize}
  \item for every $i$ with $1 \leq i \leq k$:
    \begin{itemize}
    \item one agent $n_{i}^{\rightarrow}$ with preference list
      $P_V^{\rightarrow}(i)$ and
    \item one agent $n_{i}^{\leftarrow}$ with preference list $P_V^{\leftarrow}(i)$.
    \end{itemize}
  \item a set $N_V$ of $\binom{k}{2}(2m-1)+k(2n+1)+1$ agents with preference
    list $P_{\textup{VAL}}$.
  \end{itemize}
  This completes the construction of the instance $I$. Clearly the
  given reduction can be achieved in polynomial-time. Moreover,
  since $I$ has exactly $\binom{k}{2}+k$ activities and exactly
  $2k+1$ distinct types of preference lists, both parameters
  are bounded by a function of $k$, as required. It remains to show
  that $G$ has a $k$-clique if and only if $I$ has a stable
  assignment.

  Towards showing the forward direction let $C=\{v_1,\dotsc,v_k\}$ be
  a $k$-clique of $G$ such that $v_i \in V_i$ for every $i$ with
  $1\leq i \leq k$ and for every $i$ and $j$ with $1 \leq i < j \leq
  k$ let $e_{i,j}$ be the edge between $v_i$ and $v_j$ in $G$.
  We claim that the assignment $\pi:N \rightarrow A^*$ defined in the
  following is a stable assignment for $I$. We set:
  \begin{itemize}
  \item $\pi(n_{i}^{\rightarrow})=\pi(n_{i}^{\leftarrow})=a_{i}$
    and
  \item for every $i$ and $j$ with $1 \leq i < j \leq k$, $\pi$
    assigns exactly $\alpha_{i,j}(e_{i,j})$ agents from $N_V$ to activity
    $a_{i,j}$,
  \item for every $i$ with $1 \leq i \leq k$, $\pi$
    assigns exactly $\alpha_i(v_i)-2$ agents from $N_V$ to activity
    $a_{i}$,
  \item all remaining agents (which are only in $N_V$) are assigned to $a_\emptyset$.
  \end{itemize}
  Note that
  $|\pi^{-1}(a_{i,j})|=\alpha_{i,j}(e_{i,j})$ and $|\pi^{-1}(a_i)|=\alpha_i(v_i)-2+2=\alpha_i(v_i)$ for every
  $i$ and $j$ with $1 \leq i < j \leq k$. 
  
  Towards showing that the
  assignment $\pi$ is stable, we consider any agent $n$ and distinguish the following
  cases:
  \begin{itemize}
  \item if $n$ is one of $n_{i}^{\rightarrow}$ or $n_{i}^{\leftarrow}$ for
    some $i$ with $1 \leq i \leq k$, then 
    $\pi$ is stable w.r.t. to $n$ because for every $j$ with $1\leq j
    \leq k$ and $j\neq i$,
    the edge $e_{o(i,j)}$ is incident with $v_i$.
  \item if $u \in N_V$, we consider the following cases:
    \begin{itemize}
    \item ($\pi(u)=a_i$) In this case the assignment is stable w.r.t. $n$
      because the tuple $(a_i,|\pi^{-1}(a_i)|)=(a_i,\alpha_i(v_i))$ is in the largest
      equivalence class of $P_V$.
    \item ($\pi(u)=a_{i,j}$) In this case the assignment is stable
      w.r.t. $n$
      because the tuple $(a_{i,j},|\pi{-1}(a_{i,j})|)=(a_{i,j},\alpha_{i,j}(e_{i,j}))$ is in the largest
      equivalence class of $P_V$.
    \item ($\pi(u)=a_\emptyset$) In this case the assignment is stable
      w.r.t. $n$
      because the tuples $(a_i,|\pi^{-1}(a_i)|+1)=(a_i,\alpha_i(v_i)+1)$ and
      $(a_{i,j},|\pi^{-1}(a_{i,j})|+1)=(a_{i,j},\alpha_{i,j}(e_{i,j})+1)$ are less preferred than the tuple
      $(a_\emptyset,1)$ in the preference list $P_V$ for $n$ (for
      every $i$ and $j$ with $1 \leq i < j \leq k$).
    \end{itemize}
  \end{itemize}

    Towards showing the backward direction, we start by formalizing the
  intuition given above about the preference lists $P_V^\rightarrow(i)$ and
  $P_V^\rightarrow(i)$.
  
  \begin{CLM}\label{clm:hard-gasp-pref}
    Let $i$ be an integer with $1 \leq i \leq k$ and let $\pi$ be a
    stable assignment for $I$ satisfying:
    \begin{itemize}
    \item[(A1)] $|\pi^{-1}(a_i)| \in \SB \alpha_i(v)\SM v \in V_i\SE$ and
    \item[(A2)] for every $j$ with $1 \leq j \leq k$ and $j \neq i$
      it holds that $|\pi^{-1}(a_{o(i,j)})|\in \SB
      \alpha_{o(i,j)}(e)\SM e \in E_{o(i,j)}\SE$.
    \end{itemize}
    Then the following holds for $\pi$:
    \begin{itemize}
    \item[(C1)] $\pi(n_i^\rightarrow)=\pi(n_i^\leftarrow)=a_i$. 
    \item[(C2)] For every $j$ with $1 \leq j \leq k$ and $j \neq i$, the
      unique edge $e_{o(i,j)} \in E_{o(i,j)}$ with
      $\alpha_{o(i,j)}(e)=|\pi^{-1}(a_{o(i,j)})|$ is incident
      with the unique vertex $v \in V_i$ such that $\alpha(v)=|\pi^{-1}(a_{o(i,j)})|$.
    \end{itemize}
  \end{CLM}
  Towards showing (C1) assume for a contradiction that this is not
  the case, i.e., one of the agents $n_i^\rightarrow$ or
  $n_i^\leftarrow$, in the following denoted by $n$ is not assigned to $a_i$.
  Because $(a_i,\alpha(v)+1) \in C_V^{+1}(i)$ and
  $C_V^{+1}(i)$ only contains alternatives for activity $a_i$,
  the agent $n$ would prefer to change from his current activity
  to activity $a_i$ contradicting our assumption that $\pi$ is
  stable.

  Towards showing (C2) assume for a contradiction that this is not the
  case and let $j$ be an index witnessing this, i.e., $v \notin
  e_{o(i,j)}$. Let $v'$ be the vertex in $V_i$ that is incident with
  $e_{o(i,j)}$. We distinguish two analogous cases (note that $v'\neq
  v$): (1) $\alpha(v')>\alpha(v)$ and (2) $\alpha(v')<\alpha(v)$.
  In the former case $\pi$
  would not be stable because the agent $n_i^\rightarrow$ would prefer
  to join activity $a_{i,j}$ over his current activity $a_i$; this is
  because $\alpha(v')>\alpha(v)$ and hence the equivalence class
  $C_I(i,v')$, which contains the tuple
  $(a_{o(i,j)},\alpha_{o(i,j)}(e)+1)$, is more preferred in
  $P_V^\rightarrow(i)$ than the equivalence class $C_I(i,v)$,
  which contains the tuple $(a_i,\alpha_i(v))$.
  The proof for the latter case
  is analogous, using the agent $n_i^\leftarrow$ instead of the agent $n_i^\rightarrow$.

  Note that once we show that the assumptions (A1) and (A2) hold for
  every $i$ with $1 \leq i \leq k$, property (C2) ensures that the
  vertices $\alpha^{-1}(|\pi^{-1}(a_1)|),\dotsc,
  \alpha^{-1}(|\pi^{-1}(a_k)|)$ form a $k$-clique in $G$.
  This is achieved with the following claim.
  \begin{CLM}
    For every stable assignment $\pi: N \rightarrow A^*$ for $I$, it holds
    that:
    \begin{itemize}
    \item[(A0)] $\pi(u)=a_\emptyset$ for at least one agent $u \in N_V$.
    \item[(A1)] $\pi^{-1}(a_i)\in \SB \alpha_i(v) \SM v \in V_i\SE$ for every $i$ with $1
      \leq i \leq k$.
    \item[(A2)] $\pi^{-1}(a_{i,j})\in \SB \alpha_{i,j}(e) \SM e \in E_{i,j}\SE$ for every $i$ and $j$
      with $1 \leq i < j \leq k$.      
    \end{itemize}
  \end{CLM}
Towards showing (A0), assume for a
  contradiction that this is not the case, i.e., all
  $\binom{k}{2}(2m-1)+k(2n+1)+1$ agents in $N_V$ are assigned to one
  of the $\binom{k}{2}+k$ activities in $A$. Then there either exists
  an activity $a_{i,j}$ such that more than $2m-1$ agents in $N_V$ are
  assigned to $a_{i,j}$ by $\pi$ or there exists an activity $a_i$
  such that more than $2n+1$ agents in $N_V$ are assigned to $a_i$
  by $\pi$. In the former case let $u \in N_V$ be an agent with
  $\pi(u)=a_{i,j}$. Then the assignment is not stable because $n$
  would prefer being assigned to $a_\emptyset$ over its current
  assignment to $a_{i,j}$. The latter case is analogous. This completes
  the proof for (A0).

  Because of (A0) there is at least one agent $u \in N_V$ such that
  $\pi(u)=a_\emptyset$. Hence because of the preference list $P_V$ for
  $n$, we obtain that $|\pi^{-1}(a_i)| \notin \{0,1,2\}\cup \SB
  \alpha_i(v)-1\SM v \in V_i\SE$, since otherwise $n$ would prefer
  activity $a_i$ over $a_\emptyset$. It follows that either
  $|\pi^{-1}(a_i)| \in \SB \alpha_i(v) \SM v \in V_i\SE$ or
  $|\pi^{-1}(a_i)|>\max \SB \alpha_i(v)\SE v \in V_i\SE=2n+1$. In the
  former case (A1) holds, so assume that the latter case applies.
  Since we can assume w.l.o.g. that $2n+1>2k$ (and there are only $2k$
  agents in $N \setminus N_V$), we obtain that there is at least one
  agent $u \in N_V$ such that $\pi(u)=a_i$. But then because of the
  preference list $P_V$ of $n$, $n$ would prefer activity
  $a_\emptyset$ over $a_i$, contradicting our assumption that $\pi$ is
  a stable assignment. The completes the proof of (A1).
  The proof of (A2) is analogous to the proof of (A1).
\end{proof}

\section{Result 5: Lower Bound for \cGasp{}}
\label{sec:res5}

From our previous result in conjunction with the equivalence between \Gasp{} and \cGasp{} on complete networks, we can immediately conclude that $\cGasp{}$ is also \W{1}-hard parameterized by the number of agent types and the number of activities. However, here we strengthen this result by showing that the problem remains \W{1}-hard even when one additionally parameterizes by the vertex cover number of the network. As noted in the introduction, this also implies the \W{1}-hardness of the problem when parameterized by the treewidth of the network, a question raised in previous work~\cite{GuptaRoySaurabhZehavi17}; in fact, the presented lower-bound result not only shows the (conditional) fixed-parameter intractability of the problem with a more restrictive graph parameter, but also when additionally parameterizing by the number of agent types.

\sv{\begin{THE}[$\star$]}
  \lv{\begin{THE}}\label{the:hard-cgasp}
  \cGasp{} is \Weft\emph{[1]}\hy hard parameterized by the number of activities,
  the number of agent types, and the vertex cover number of the network.
\end{THE}
\begin{proof}
  \sloppypar The proof is via a parameterized reduction from \mc{}, i.e.,
  given an instance $(G,k)$ of \mc{} with partition $V_1,\dotsc,V_k$,
  we construct an equivalent instance $I=(N,A,(\pref_n)_{n\in N},L)$ of \cGasp{} in
  polynomial time with $\binom{k}{2}+k$ activities,
  $\binom{k}{2}+3k$ agent types, and whose network $(N,L)$ has vertex
  cover number at most $\binom{k}{2}+2k$.

  The main ideas behind the reduction are quite similar to the
  reduction used in the proof of Theorem~\ref{the:hard-gasp}. The main
  differences is that to ensure that the network $(N,L)$ has a small
  vertex cover number, it is necessary to split the set $N_V$ of
  agents used in the previous reduction, into sets $N_i$ and $N_{i,j}$
  for every $i$ and $j$ with $1 \leq i <j \leq k$ such that the agents
  in a set $N_i$ can only be assigned to activity $a_i$ (or
  $a_\emptyset$) and the agents in a set $N_{i,j}$ can only be
  assigned to activity $a_{i,j}$ (or
  $a_\emptyset$). This way the agents $n_i^\rightarrow$ and
  $n_i^\leftarrow$ only need to be connected to agents in $N_i$ and
  $N_{i,j}$ (for any $j \neq i$) but not with all agents in $N_V$.

  Let $I'$ be the instance of \Gasp{} as defined in the proof of
  Theorem~\ref{the:hard-gasp}. Since the instance $I$ is defined quite
  similar to the instance $I'$, we will refer to $I'$ for the
  construction of $I$. In particular, $I$ has the same activities as
  $I'$, i.e., $I$ has one activity $a_i$ for every $i$ with $1 \leq i
  \leq k$ and one activity $a_{i,j}$ for every $i$ and $j$ with $1
  \leq i <j \leq k$. 
  For every $i$ and $j$ with $1 \leq i < j \leq k$ let $\alpha_i$ be a
  bijection from $V_i$ to the set $\{3,5,\dotsc, 2(n-1)+1,2n+1\}$
  and similarly let $\alpha_{i,j}$ be a bijection
  from $E_{i,j}$ to the set $\{3,5,\dotsc, 2m+1\}$. Note that
  $\alpha_i$ and $\alpha_{i,j}$ are defined almost the same as in the proof of
  Theorem~\ref{the:hard-gasp}, the only difference being the
  definition of the image $\alpha_{i,j}$, which is now
  $\{3,5,\dotsc, 2m+1\}$ instead of $\{1,3,\dotsc, 2m-1\}$.


  For the definition of the preference lists, we will mainly use the
  equivalence classes defined in the proof of Theorem~\ref{the:hard-gasp}, i.e.,
  the classes $C_V(i)$, $C_V^{+1}$, $C_E(i,j)$, and $C_I(i,v)$. Apart
  from those we will also need a slightly modified version of the
  equivalence class $C_V^{+1}(i)$, which we denote by $C_{V}{+1,2}(i)$
  and set to $\{(a_i,2)\}\cup C_V^{+1}(i)$. We are now ready to define
  the preference lists, required by the reduction.
  \begin{itemize}
  \item For every $i$ with $1 \leq i \leq k$, let $v_1,\dotsc,v_u$ be the
    unique ordering of the vertices in $V_i$ in ascending order
    w.r.t. $\alpha_i$. We define the following two preference lists
    for every $i$ with $1 \leq i \leq k$:
    \begin{itemize}
    \item The \emph{vertex-validity} preference
      list, denoted by $P_{\textup{VAL}}^V(i)$, defined as $C_V(i) > (a_i,2) >
      (a_i,1) > (a_\emptyset,1)$.
    \item The \emph{forward-vertex} preference list, denoted by $P_V^{\rightarrow}(i)$,
      defined as $C_{V}^{+1,2}(i) > C_I(i,v_u) > C_I(i,v_{n-1}) > \dotsb > C_{I}(i,v_1) >
        (a_\emptyset,1)$.
    \item The \emph{backward-vertex} preference list, denoted by $P_V^{\leftarrow}(i)$,
      defined as $C_{V}^{+1,2}(i) > C_I(i,v_1) > C_I(i,v_2) > \dotsb > C_I(i,v_u) >
      (a_\emptyset,1)$.
    \end{itemize}
  \item 
    For every $i$ and $j$ with $1 \leq i < j \leq k$ the \emph{edge-validity} preference
    list, denoted by $P^E_{\textup{VAL}}(i,j)$, and defined as $C_E(i,j) > (a_{i,j},2) >
    (a_{i,j},1) > (a_\emptyset,1)$.
  \end{itemize}

  We are now ready to define the set $N$ of agents:
  \begin{itemize}
  \item for every $i$ with $1 \leq i \leq k$:
    \begin{itemize}
    \item one agent $n_{i}^{\rightarrow}$ with preference list
      $P_V^{\rightarrow}(i)$,
    \item one agent $n_{i}^{\leftarrow}$ with preference list
      $P_V^{\leftarrow}(i)$, and 
    \item a set $N_i$ of $2n+3$ agents with preference list $P_{\textup{VAL}}^V(i)$.
    \end{itemize}
  \item for every $i$ and $j$ with $1 \leq i < j \leq k$, 
    a set $N_{i,j}$ of $2m+3$ agents with preference list
    $P^E_{\textup{VAL}}(i,j)$. In the following let $n_{i,j}$ be one of the agents in
    $N_{i,j}$.
  \end{itemize}
  Finally, the links $L$ between the agents are given by:
  \begin{itemize}
  \item for every $i$ with $1 \leq i \leq k$:
    \begin{itemize}
    \item a link between $n_i^{\rightarrow}$ and $n_i^{\leftarrow}$,
    \item for every $u \in N_i$ a link between $n_i^{\rightarrow}$ and
      $n$,
    \item for every $j$ with $1\leq j \leq k$ and $j \neq k$ a link
      between $n_i^\rightarrow$ and $n_{o(i,j)}$,
    \end{itemize}
  \item for every $i$, $j$ with $1 \leq i <j \leq k$ and every $u \in N_{i,j}\setminus \{n_{i,j}\}$
     a link between $n_{i,j}$ and $n$.
  \end{itemize}
  An illustration of the network $(N,L)$ as defined above is given in
  Figure~\ref{fig:hardness-network}.
  
  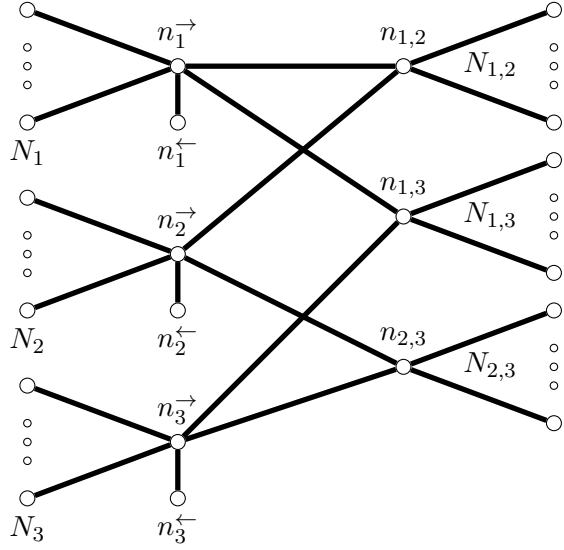
\begin{figure}[tb]
    \centering
    \begin{tikzpicture}
      \tikzstyle{every node}=[]
      \tikzstyle{gn}=[circle, inner sep=2pt,draw, node distance=2cm]
      \tikzstyle{dots}=[circle, inner sep=1pt,draw, node distance=2cm]
      \tikzstyle{every edge}=[draw, line width=2pt]

      \draw
      node[gn, label=above:$n_1^\rightarrow$] (n1r) {}
      node[dots, node distance=2cm, left of=n1r] (n1d2) {}
      node[dots, node distance=0.25cm, above of=n1d2] (n1d1) {}
      node[dots, node distance=0.25cm, below of=n1d2] (n1d3) {}
      node[gn, node distance=0.5cm, above of=n1d1] (n1r1) {}
      node[gn, node distance=0.5cm, below of=n1d3, label=below:$N_1$] (n1r2) {}
      
      node[gn, node distance=0.75cm, below of=n1r,
      label=below:$n_1^\leftarrow$] (n1l) {}

      (n1r) edge (n1r1)
      (n1r) edge (n1r2)
      (n1r) edge (n1l)
      ;

      \draw
      node[gn, node distance=3cm, right of=n1r, label=above:$n_{1,2}$] (n12) {}
      node[dots, node distance=2cm, right of=n12, label={[label distance=0.3cm]left:$N_{1,2}$}] (n12d2) {}
      node[dots, node distance=0.25cm, above of=n12d2] (n12d1) {}
      node[dots, node distance=0.25cm, below of=n12d2] (n12d3) {}
      node[gn, node distance=0.5cm, above of=n12d1] (n12s1) {}
      node[gn, node distance=0.5cm, below of=n12d3] (n12s2) {}
      
      (n12) edge (n12s1)
      (n12) edge (n12s2)
      ;

      \draw
      node[gn, node distance=2.5cm, below of=n1r,label=above:$n_2^\rightarrow$] (n2r) {}
      node[dots, node distance=2cm, left of=n2r] (n2d2) {}
      node[dots, node distance=0.25cm, above of=n2d2] (n2d1) {}
      node[dots, node distance=0.25cm, below of=n2d2] (n2d3) {}
      node[gn, node distance=0.5cm, above of=n2d1] (n2r1) {}
      node[gn, node distance=0.5cm, below of=n2d3, label=below:$N_2$] (n2r2) {}
      
      node[gn, node distance=0.75cm, below of=n2r,
      label=below:$n_2^\leftarrow$] (n2l) {}

      (n2r) edge (n2r1)
      (n2r) edge (n2r2)
      (n2r) edge (n2l)
      ;

      \draw
      node[gn, node distance=2cm, below of=n12, label=above:$n_{1,3}$] (n13) {}
      node[dots, node distance=2cm, right of=n13, label={[label distance=0.3cm]left:$N_{1,3}$}] (n13d2) {}
      node[dots, node distance=0.25cm, above of=n13d2] (n13d1) {}
      node[dots, node distance=0.25cm, below of=n13d2] (n13d3) {}
      node[gn, node distance=0.5cm, above of=n13d1] (n13s1) {}
      node[gn, node distance=0.5cm, below of=n13d3] (n13s2) {}
      
      (n13) edge (n13s1)
      (n13) edge (n13s2)
      ;

      \draw
      node[gn, node distance=2.5cm, below of=n2r,label=above:$n_3^\rightarrow$] (n3r) {}
      node[dots, node distance=2cm, left of=n3r] (n3d2) {}
      node[dots, node distance=0.25cm, above of=n3d2] (n3d1) {}
      node[dots, node distance=0.25cm, below of=n3d2] (n3d3) {}
      node[gn, node distance=0.5cm, above of=n3d1] (n3r1) {}
      node[gn, node distance=0.5cm, below of=n3d3, label=below:$N_3$] (n3r2) {}
      
      node[gn, node distance=0.75cm, below of=n3r,
      label=below:$n_3^\leftarrow$] (n3l) {}

      (n3r) edge (n3r1)
      (n3r) edge (n3r2)
      (n3r) edge (n3l)
      ;

      \draw
      node[gn, node distance=2cm, below of=n13, label=above:$n_{2,3}$] (n23) {}
      node[dots, node distance=2cm, right of=n23, label={[label distance=0.3cm]left:$N_{2,3}$}] (n23d2) {}
      node[dots, node distance=0.25cm, above of=n23d2] (n23d1) {}
      node[dots, node distance=0.25cm, below of=n23d2] (n23d3) {}
      node[gn, node distance=0.5cm, above of=n23d1] (n23s1) {}
      node[gn, node distance=0.5cm, below of=n23d3] (n23s2) {}
      
      (n23) edge (n23s1)
      (n23) edge (n23s2)
      ;

      \draw
      (n1r) edge (n12)
      (n1r) edge (n13)
      (n2r) edge (n12)
      (n2r) edge (n23)
      (n3r) edge (n13)
      (n3r) edge (n23)
      ;
    \end{tikzpicture}
    \caption{An illustration of the network $(N,L)$ obtained in the
      reduction of Theorem~\ref{the:hard-cgasp} for the case that $k=3$.}
    \label{fig:hardness-network}
  \end{figure}
  
  This completes the construction of the instance $I$. Observe that
  the set $\SB n_i^\rightarrow, n_i^\leftarrow \SM 1 \leq i \leq k\SE
  \cup \SB n_{i,j} \SM 1 \leq i < j \leq k\SE$ is a vertex cover of
  the network $(N,L)$ of size at most $\binom{k}{2}+2k$, and hence the
  network has vertex cover number at most $\binom{k}{2}+2k$.
Clearly the
  given reduction can be achieved in polynomial-time. Moreover,
  since $I$ has exactly $\binom{k}{2}+k$ activities, exactly
  $\binom{k}{2}+3k$ distinct types of preference lists, and the vertex
  cover number of the network $(N,L)$ is at most $\binom{k}{2}+2k$, all parameters
  are bounded by a function of $k$, as required. It remains to show
  that $G$ has a $k$-clique if and only if $I$ has a stable
  assignment.
Towards showing the forward direction let $C=\{v_1,\dotsc,v_k\}$ be
  a $k$-clique of $G$ such that $v_i \in V_i$ for every $i$ with
  $1\leq i \leq k$ and for every $i$ and $j$ with $1 \leq i < j \leq
  k$ let $e_{i,j}$ be the edge between $v_i$ and $v_j$ in $G$.
  We claim that the assignment $\pi:N \rightarrow A^*$ defined in the
  following is a stable assignment for $I$. We set:
  \begin{itemize}
  \item $\pi(n_{i}^{\rightarrow})=\pi(n_{i}^{\leftarrow})=a_{i}$,
  \item for every $i$ with $1 \leq i \leq k$, $\pi$
    assigns exactly $\alpha_i(v_i)-2$ agents from $N_i$ to activity
    $a_{i}$,
  \item for every $i$ and $j$ with $1 \leq i < j \leq k$,
    $\pi(n_{i,j})=a_{i,j}$, and
  \item for every $i$ and $j$ with $1 \leq i < j \leq k$, $\pi$
    assigns exactly $\alpha_{i,j}(e_{i,j})-1$ agents from $N_{i,j}$ to activity
    $a_{i,j}$,
  \item all remaining agents are assigned to $a_\emptyset$.
  \end{itemize}
  Note that for every
  $i$ and $j$ with $1 \leq i < j \leq k$ the agents assigned to
  activities $a_i$ and $a_{i,j}$ are connected and moreover
  $|\pi^{-1}(a_{i,j})|=\alpha_{i,j}(e_{i,j})$ and $|\pi^{-1}(a_i)|=\alpha_i(v_i)$.
  Let $n$ be an agent, we consider the following
  cases:
  \begin{itemize}
  \item if $n$ is one of $n_{i}^{\rightarrow}$ or $n_{i}^{\leftarrow}$ for
    some $i$ with $1 \leq i \leq k$, then
    the assignment $\pi$ is stable w.r.t. to $n$ because for every $j$ with $1\leq j
    \leq k$ and $j\neq i$, the edge
    $e_{o(i,j)}$ is incident
    to $v_{i}$ in $G$ and hence the tuples $(a_i,\alpha_i(v_i))$ and
    $(a_{o(i,j)},\alpha_{o(i,j)}(e_{o(i,j)})+1)$ are in the same
    equivalence class of $P_V^{\rightarrow}$ and $P_V^{\leftarrow}$.
  \item if $u \in N_i$ and $\pi(u)=a_i$, then
    the assignment $\pi$ is stable w.r.t. $n$ because the tuple $(a_i,\alpha_i(v_i))$
    is in the most preferred equivalence class of $P^V_{\textup{VAL}}(i)$.
  \item if $u \in N_i$ and $\pi(u)=a_\emptyset$, then
    the assignment $\pi$ is stable w.r.t. $n$ because $(a_\emptyset,1)$ is preferred to
    $(a_i,\alpha_i(v_i)+1)$ as well as to any tuple with any other activity in $P^V_{\textup{VAL}}(i)$.
  \item if $u \in N_{i,j}$ and $\pi(u)=a_{i,j}$, then
    the assignment $\pi$ is stable w.r.t. $n$ because the the tuple
    $(a_{i,j},\alpha_{i,j}(e_{i,j}))$ is in the most preferred
    equivalence class of $P^E_{\textup{VAL}}(i,j)$.
  \item if $u \in N_{i,j}$ and $\pi(u)=a_\emptyset$, then
    the assignment $\pi$ is stable w.r.t. $n$ because
    $(a_\emptyset,1)$ is preferred to
    $(a_{i,j},\alpha_{i,j}(e_{i,j})+1)$ as well as to any tuple with any other activity in $P^E_{\textup{VAL}}(i,j)$.
  \end{itemize}

  The reverse direction follows immediately from the following claim.
  \begin{CLM}
    For every stable assignment $\pi: N \rightarrow A^*$ for $I$, it holds
    that:
    \begin{itemize}
    \item[(C1)] for every $i$ with $1 \leq i \leq k$, there is an
      agent $u \in N_i$ such that $\pi(u)=(a_\emptyset,1)$,
    \item[(C2)] for every $i$ and $j$ with $1 \leq i < j \leq k$,
      there is an agent $u \in N_{i,j}$ such that
      $\pi(u)=(a_\emptyset,1)$,
    \item[(C3)] $|\pi^{-1}(a_i)|\in \SB\alpha_i(v)\SM v \in V_i\SE$ for every $i$ with $1
      \leq i \leq k$,
    \item[(C4)]
      $\pi(n_{i}^{\rightarrow})=\pi(n_{i}^{\leftarrow})=a_{i}$
      for every $i$  with $1 \leq i \leq k$,
    \item[(C5)] $|\pi^{-1}(a_{i,j})|\in \SB \alpha_{i,j}(e)\SM e \in E_{i,j}\SE$ for every $i$ and $j$
      with $1 \leq i < j \leq k$,      
    \item[(C6)] $\pi(n_{i,j})=a_{i,j}$
      for every $i$ and $j$  with $1 \leq i <j \leq k$,
    \item[(C7)] $\alpha_i^{-1}(|\pi^{-1}(a_i)|) \in
      \alpha_{o(i,j)}^{-1}(|\pi^{-1}(a_{o(i,j)})|)$ for every $i$ and
      $j$ with $1\leq i,j \leq k$ and $i \neq j$.
    \end{itemize}
  \end{CLM}
Towards showing (C1) suppose for a
  contradiction that this is not the case. Then because all agents in
  $N_i$ must either be assigned to $a_i$ or $a_\emptyset$ by $\pi$, we
  obtain that all $2n+3$ agents  in
  $N_i$ must be assigned to $a_i$. However
  such an assignment would not be stable because any tuple $(a_i,x)$
  with $x\geq 2n+3$ is less preferred than $(a_\emptyset,1)$ in every
  preference list.

  The proof of (C2) is very similar to the proof of (C1). Namely, suppose for a
  contradiction that this is not the case. Then all $2m+3$ agents in
  $N_{i,j}$ must be assigned to $a_{i,j}$, however
  such an assignment would not be stable because any tuple $(a_{i,j},x)$
  with $x\geq 2m+3$ is less preferred than $(a_\emptyset,1)$ in every
  preference list.

  Towards showing (C3) we first show that $|\pi^{-1}(a_i)|>1$.
  Because of (C1) there is an agent $u \in N_i$ with
  $\pi(u)=\emptyset$. Hence $|\pi^{-1}(a_i)|\neq 0$ since otherwise
  the agent $n$ would prefer to switch to $a_i$. Now suppose for a
  contradiction that $|\pi^{-1}(a_i)|=1$. Because the agents in $N_i$
  are the only agents that prefer the tuple $(a_i,1)$ over
  $(a_\emptyset,1)$, it holds that $\pi^{-1}(a_i) \subseteq N_i$. But
  since the tuple $(a_i,2)$ is a tuple that is in the highest
  equivalence class $C_V^{+1,2}(i)$ of the preference list $P_i^\rightarrow$
  for $n_i^\rightarrow$ and $n_i^\rightarrow$ is linked with every
  vertex in $N_i$, the agent $n_i^\rightarrow$ would prefer to switch
  to $a_i$, contradicting our assumption that $\pi$ is
  stable. Consequently $|\pi^{-1}(a_i)|>1$ and we show next that
  $\pi^{-1}(a_i)$ contains $n_i^\rightarrow$. Observe that the agents in
  $\{n_i^\leftarrow,n_i^\rightarrow\}\cup N_i$ are the only agents in
  $N$ that can be assigned to $a_i$; all other agents prefer the tuple
  $(a_\emptyset,1)$ over any tuple involving the activity $a_i$.
  Since $|\pi^{-1}(a_i)|>1$ the set $\pi^{-1}(a_i)$ can only be connected if it contains
  $n_i^\rightarrow$. Hence we have $|\pi^{-1}(a_i)|>1$ and
  $n_i^\rightarrow \in \pi^{-1}(a_i)$. Because of (C1) there is an
  agent $u \in N_i$ with $\pi(u)=a_\emptyset$. Since $n$ is linked
  with $n_i^\rightarrow$ and prefers $a_i$ over $a_\emptyset$,
  whenever $|\pi^{-1}(a_i)| \in \{0,1\}\cup \SB
  \alpha_i(v)-1\SM v \in V_i\SE=\{0,1,2,4,6, \dotsc, 2n\}$, we
  obtain that either $|\pi^{-1}(a_i)| \in \SB \alpha_i(v) \SM v \in V_i\SE=\{3,5, \dotsc, 2n+1\}$ or
  $|\pi^{-1}(a_i)|>2n+1$. In the former case (C3) holds, so assume
  that the latter case applies.
  Since the agents in $\{n_i^\leftarrow,n_i^\rightarrow\}\cup N_i$ are
  the only agents in $N$ that can be assigned to $a_i$ and we can assume w.l.o.g. that
  $2n+1>2$, we obtain that there is at least one agent $u \in N_i$
  such that $\pi(u)=a_i$. But then because of the
  preference list $P_V(i)$ of $n$, $n$ would prefer activity
  $a_\emptyset$ over $a_i$, contradicting our assumption that $\pi$ is
  a stable assignment. The completes the proof of (C3).

  Towards showing (C5) first observe that the agents in
  $\SB n_i^\leftarrow,n_i^\rightarrow\SM 1\leq i \leq k\}\cup N_{i,j}$ are the only agents in
  $N$ that can be assigned to $a_{i,j}$; all other agents prefer the tuple
  $(a_\emptyset,1)$ over any tuple involving the activity $a_{i,j}$.
  Moreover because of (C4), we obtain that only the agents in
  $N_{i,j}$ can actually be assigned to $a_{i,j}$. Moreover because
  all agents in $N_{i,j}$ must either be assigned to $a_{i,j}$ or to
  $a_\emptyset$ and due to (C2) there is always an agent $n' \in
  N_{i,j}$ with $\pi(n')=a_\emptyset$, and since the set $N_{i,j}$ is
  connected by $L$, we obtain that there is always an agent $n \in
  N_{i,j}$ with $\pi(u)=a_\emptyset$ that is linked to an agent in
  $\pi^{-1}(a_{i,j})$. Hence it follows from the preference list
  $P_E(i,j)$ of the agents in $N_{i,j}$ that $|\pi^{-1}(a_{i,j})|\notin
  \{0,1\}\cup \SB \alpha_{i,j}(e)-1 \SM e \in E_{i,j}\SE$, since otherwise the agent
  $n$ would prefer to switch to $a_{i,j}$. 
  Hence either $|\pi^{-1}(a_{i,j})| \in \SB \alpha_{i,j}(e) \SM e \in E_{i,j}\SE=\{3,5, \dotsc, 2m+1\}$ or
  $|\pi^{-1}(a_{i,j})|>2m+1$. In the former case (C5) holds, so assume
  that the latter case applies. Then there is an agent $u \in N_{i,j}$
  with $\pi(u)=a_{i,j}$, but since $|\pi^{-1}(a_{i,j})|>2m+1$ the
  agent would prefer to be assigned to $a_\emptyset$, contradicting
  our assumption that the assignment $\pi$ is stable. This completes
  the proof of (C5).

  (C6) can be obtained as follows. Because of
  (C4), we have that $|\pi^{-1}(a_{i,j})|\geq
  3$. Since (as observed already in the proof of (C5)) only the agents
  in $N_{i,j}$ can be assigned to activity $a_{i,j}$, we obtain that
  $\pi(n_{i,j})=a_{i,j}$ since otherwise $\pi^{-1}(a_{i,j})$ would not
  be connected.
  
  Towards showing (C7) assume for a contradiction that there are $i$
  and $j$ with $1 \leq i, j \leq k$ and $i\neq j$ such that $v \notin
  e$, where $v=\alpha_i^{-1}(|\pi^{-1}(a_i)|)$ and
  $e=\alpha_{o(i,j)}(|\pi^{-1}(a_{o(i,j)})$. Observe first that
  because of (C3) and (C4) $v$ and $e$ are properly defined and
  moreover $v \in V_i$ and $e \in E_{i,j}$. Let $v'$ be the endpoint
  of $e$ in $V_i$, which because $v \notin e$ is not equal to $v$.
  We distinguish two analogous cases: (1) $\alpha_i(v')<\alpha_i(v)$
  and (2) $\alpha_i(v') > \alpha_i(v)$. In the former case consider
  the agent $n_i^{\leftarrow}$. Because of (C4) and (C6), it holds that
  $\pi(n_i^{\leftarrow})=a_i$ and $\pi(n_{i,j})=a_{o(i,j)}$, which
  implies that $n_i^\leftarrow$ is linked with an agent, namely the
  agent $n_{i,j}$ assigned to $a_{o(i,j)}$ by $\pi$. Together with
  the facts that the equivalence class
  $C_I(i,v')$ contains the tuple $(a_{o(i,j)},\alpha(e))$, the
  equivalence class $C_{I}(i,v)$ contains the tuple $(a_i,\alpha(v))$,
  and $C_{E}(i,v')$ is preferred over $C_{E}(i,v)$ in the preference list
  for $n_i^{\leftarrow}$, we obtain that the agent $n_i^{\leftarrow}$
  prefers the activity $a_{o(i,j)}$ over its current activity $a_i$,
  contradicting the stability of $\pi$. The proof for the latter case 
  is analogous using the agent $n_i^{\rightarrow}$ instead of the
  agent $n_i^{\leftarrow}$. This completes the proof for (C7).
\end{proof}

\section{Conclusion}

We obtained a comprehensive picture of the parameterized complexity of Group
Activity Selection problems parameterized by the number of agent
types, both with and without the number of activities as an additional parameter. Our positive results suggest that using the number of agent types is a highly appealing
parameter for \Gasp{} and its variants; indeed, this parameter will
often be much smaller than the number of agents due to the way
preference lists are collected or estimated (as also argued in initial
work on \Gasp{}~\cite{DarmannEKLSW12}). For instance, in the
large-scale event management setting of \Gasp{} (or \sGasp{}), one
would expect that preference lists for event participants are
collected via simple questionnaires---and so the number of agent types
would remain fairly small regardless of the size of the event.

We believe that the techniques used to obtain the presented results, and especially the Subset Sum tools developed to this end, are of broad interest to the algorithms community. For instance, \textsc{Multidimensional Subset Sum} (\textsc{MSS}) has been used as a starting point for \W{1}-hardness reductions in at least two different settings over the past year~\cite{GanianOrdyniakRamanujan17,GanianKO18}, but the simple and partitioned variant of the problem (i.e., \textsc{SMPSS}) is much more restrictive and hence forms a strictly better starting point for any such reductions in the future. This is also reflected in our proof of the \W{1}-hardness of \textsc{SMPSS}, which is \emph{significantly} more involved than the analogous result for \textsc{MSS}. Likewise, we expect that the developed algorithms for \textsc{Tree Subset Sum} and \textsc{Multidimensional Partitioned Subset Sum} may find applications as subroutines for (parameterized and/or classical) algorithms in various settings.


Note that there is
now an almost complete picture of the complexity of Group Activity
Selection problems
w.r.t. any combination of the parameters number of agents, number of
activities, and number of agent types (see also
Table~\ref{tbl:intro-res}). There is only one piece missing, namely,
the parameterized complexity of \sGasp{} parameterized by the number
of agents, which we resolve for completeness with the following
theorem.
\begin{THE}
  \sGasp{} is fixed-parameter tractable parameterized by the number of agents.
\end{THE}
\begin{proof}
  Let $I=(N,A,(P_n)_{n\in N})$ be a \sGasp{} instance.
  The main idea behind the algorithm is to guess (i.e., branch over) the set $M_\emptyset$
  of agents that are assigned to $a_\emptyset$ as well as a partition
  $\MMM$ of the remaining agents, i.e., the
  agents in $N \setminus M_\emptyset$, and then check whether there is
  a stable assignment $\pi$ for $I$ such that:
  \begin{itemize}
  \item[(P1)] $\pi^{-1}(a_\emptyset)=M_\emptyset$ and 
  \item[(P2)] $\SB \pi^{-1}(a) \SM a \in A \SE \setminus
    \{\emptyset\}=\MMM$, i.e., $\MMM$ corresponds to the grouping of
    agents into activities by $\pi$.
  \end{itemize}
  Since there are at most $n^n$ possibilities for $M_\emptyset$ and
  $\MMM$ and those can be enumerated in time $\bigoh(n^n)$, it remains to
  show how to decide whether there is a stable assignent for $I$
  satisfying (P1) and (P2) for any given $M_\emptyset$ and
  $\MMM$. Towards showing this, we first consider
  the implications for a stable assignment resulting from assigning
  the agents in $M_\emptyset$ to $a_\emptyset$. Namely, let $P_n'$ for
  every $n \in N$ be the approval set obtained from $P_n$ after
  removing all alternatives $(a,i)$ such that $i\neq 0$ and there is an agent
  $n_\emptyset \in M_\emptyset$ with $(a,i+1) \in
  P_{n_\emptyset}$. Moreover, let $A_{\neq \emptyset}$ be the set of
  all activities that cannot be left empty if the agents in
  $M_\emptyset$ are assigned to $a_\emptyset$, i.e., the set of all
  activities such that there is an agent $n_\emptyset \in M_\emptyset$
  with $(a,1) \in P_{n_\emptyset}$. Now consider a set $M \in \MMM$,
  and observe that the set $A_M$ of activities that the agents in $M$ can be
  assigned to in any stable assignment satisfying (P1) and (P2) is
  given by: $A_M=\SB a \SM |M| \in \bigcap_{n \in M}P_n(a)'\SE$.
  Let $B$ be the bipartite graph having $\MMM$ on one side and $A$ on
  the other side and having
  an edge between a vertex $M \in \MMM$ and a vertex $a \in A$ if $a
  \in A_M$.
  We claim that $I$ has a stable assignment satisfying (P1)
  and (P2) if and only if $B$ has a matching that saturates $\MMM \cup
  A_{\neq \emptyset}$. Since deciding the existence of such a matching
  can be achieved in time $\bigoh(\sqrt{|V(B)|}|E(B)|)=\bigoh(\sqrt{|N
    \cup A|}|N||A|)$ (see e.g.~\cite[Lemma
  4]{GanianOrdyniakRamanujan17}), establishing this claim is the last component required for the proof of the theorem.

  Towards showing
  the forward direction, let $\pi$ be a stable assignment for $I$
  satisfying (P1) and (P2). Then $O=\SB \{a,\phi^{-1}(a)\} \SM a \in A \SE$ is a matching in $B$ that saturates $\MMM \cup A_{\neq
    \emptyset}$. Note that $O$ saturates $\MMM$ due to (P2), moreover,
  $O$ saturates $A_{\neq \emptyset}$ since otherwise there would be an
  activity $a \in A_{\neq \emptyset}$ with $\pi^{-1}(a)=\emptyset$,
  which due to the definition of $A_{\neq \emptyset}$ and (P1) implies
  there is an agent $n$ with $\pi(n)=a_\emptyset$ such that $1 \in
  P_n(a)$, contradicting our assumption that $\pi$ is stable.

  Towards showing the reverse direction, let $O$ be a matching in $B$
  that saturates $\MMM \cup A$. Then the assignment $\pi$ mapping
  all agents in $M$ (for every $M \in \MMM$) to its partner in $O$ and
  all other agents to $a_\emptyset$ clearly already satisfies (P1)
  and (P2). It remains to show that it is also stable. Note that $\pi$
  is individually rational because of the construction of $B$. Moreover,
  assume for a contradiction that there is an agent $n \in
  N_\emptyset$ with $\pi(n)=a_\emptyset$ and an activity $a \in A$
  such that $(a,|\pi^{-1}(a)|+1) \in P_n$. If $|\pi^{-1}(a)|=0$, then $a
  \in A_{\neq \emptyset}$ and hence $|\pi^{-1}(a)|>0$ (because $O$
  saturates $A_{\neq \emptyset}$), a contradiction. If on the other hand
  $|\pi^{-1}(a)|\neq 0$, then $\{M,a\} \in O$ (for some $M \in \MMM$),
  but $(a,|\pi^{-1}(a)|) \notin P_n'$ and hence $\{M,a\} \notin E(B)$, also
  a contradiction.
\end{proof}

For future work, we believe that it would be interesting to see how
the complexity map changes if one were to consider the number of
activity types instead of the number of activities in our parameterizations.

\bibliographystyle{acm} 
\bibliography{literature} 

\end{document}